\documentclass{article}
\usepackage{setspace} 
\usepackage[margin=1in]{geometry}
\usepackage{wrapfig}
\usepackage{subfig}
\usepackage{graphicx}
\usepackage{float}
\usepackage{bbm}
\usepackage[title]{appendix}
\usepackage{mathtools}     
\usepackage{multirow}
\usepackage{natbib}
\usepackage{amsmath, amssymb, enumerate, amsthm}
\usepackage{amsfonts}
\usepackage{bigints}
\usepackage{mathrsfs}
\usepackage{scalerel}
\usepackage{xcolor}
\usepackage[colorlinks=true,
            pdfpagemode=UseNone,
            urlcolor=blue,
            citecolor=blue,
            linkcolor=blue,
            bookmarks=true,
            backref=page
            ]{hyperref}
\newcommand{\scaledN}{\scaleto{N}{4pt}}
\newcommand{\scaled}[1]{\scaleto{#1}{4pt}}

\newcommand{\ind}[1]{\mathbbm{1}(#1)}

\newcommand{\limn}{\lim_{n\rightarrow\infty }}
\newcommand{\limN}{\lim_{N\rightarrow\infty }}
\newcommand{\rd}[1]{\mathbb{R}^{#1}}
\newcommand{\re}{\mathbb{R}}
\newcommand{\dom}{[0,1]}
\newcommand{\intd}{\int_{\dom}}

\newcommand{\kerr}{\mathcal{K}}
\newcommand{\kerrO}{\mathscr{K}}
\newcommand{\ltwo}{\mathscr{L}^2}
\newcommand{\norm}[1]{\left\lVert#1\right\rVert}
\newcommand{\eqd}{\stackrel{\text{d}}{=}}

\newcommand{\sumJ}{\sum_{j=1}^J}
\newcommand{\rhatbarj}{\overline{\widehat{R}}_{j}}

\newcommand{\rankji}{\widehat{R}_{ji}}
\newcommand{\cond}{\stackrel{\text{d}}{\rightarrow}}

\newcommand{\D}[1]{{\rm D}\left(#1\right)}
\newcommand{\Dd}{{\rm D}}
\newcommand{\Var}[1]{{{\rm V}ar}\left(#1\right)}
\newcommand{\Varr}[1]{{{\rm V}ar}(#1)}
\newcommand{\E}[2]{{\rm {E}}_{#1}\left[#2\right]}
\newcommand{\Eee}[2]{{\rm {E}}_{#1}[#2]}
\DeclareRobustCommand{\rchi}{{\mathpalette\irchi\relax}}
\newcommand{\irchi}[2]{\raisebox{\depth}{$#1\chi$}} 

\DeclarePairedDelimiter\floor{\lfloor}{\rfloor}
\DeclarePairedDelimiter\ip{\langle}{\rangle}

\DeclareMathOperator{\med}{Med}

\DeclareMathOperator{\LTR}{LTR}
\DeclareMathOperator{\HD}{HD}

\DeclareMathOperator{\RP}{RPD}
\DeclareMathOperator{\MBD}{MBD}
\DeclareMathOperator{\MMBD}{MMBD}

\DeclareMathOperator{\MFHD}{MFHD}
\DeclareMathOperator{\SD}{SD}
\DeclareMathOperator{\KSD}{KSD}

\newtheorem{theorem}{Theorem}

\newtheorem{ass}{Assumption}

\title{Robust nonparametric hypothesis tests for differences in the covariance structure of functional data}
\author{Kelly Ramsay, Shojaeddin Chenouri}
\date{November 2020}

\begin{document}
\doublespacing
\maketitle
\begin{abstract}
We develop a group of robust, nonparametric hypothesis tests which detect differences between the covariance operators of several populations of functional data. 
These tests, called FKWC tests, are based on functional data depth ranks. 
These tests work well even when the data is heavy tailed, which is shown both in simulation and theoretically. 
These tests offer several other benefits, they have a simple distribution under the null hypothesis, they are computationally cheap and they possess transformation invariance properties. 
We show that under general alternative hypotheses these tests are consistent under mild, nonparametric assumptions. 
As a result of this work, we introduce a new functional depth function called $L^2$-root depth which works well for the purposes of detecting differences in magnitude between covariance kernels. 
We present an analysis of the FKWC test using $L^2$-root depth under local alternatives. 
In simulation, when the true covariance kernels have strictly positive eigenvalues, we show that these tests have higher power than their competitors, while still maintaining their nominal size. 
We also provide a methods for computing sample size and performing multiple comparisons. 
\end{abstract}
\section{Introduction}
Data such that the observations are each a smooth curve, deemed functional data, is being increasingly observed in a variety of fields. 
For example, medical images \citep{Lopez-Pintado2017, Aston2017}, intraday financial asset returns \citep{CEROVECKI2019353} and environmental ``omics'' data \citep{PINA2018583} can all be interpreted as functional data. 
As such, many functional analogues of univariate and multivariate statistical tools are needed. 
One such tool is the notion of common variance in the functional context; common covariance operator or covariance kernel. 
In this work, we introduce new, nonparametric functional $k$-sample tests for equal covariance structure. 
We call this class of tests the functional Kruskal-Wallis tests for covariance structure, or for short, FKWC tests. 

Before introducing the FKWC tests, it is useful to review the existing, related works. 
Some early related works include \citep{James2006}, who presented a test for a difference in the shape of the mean function between two populations of curves and \citep{Gabrys2007, Aue2009, Horvath2010} who all consider tests related to serial dependence, or time series characteristics.

\cite{Panaretos2010} were the first to discuss comparing the covariance structures of two functional data populations. 
They provide a two sample test and restrict their attention to that of Gaussian processes. 
The test is based on the Hilbert-Schmidt norm for integral operators. 
\cite{Fremdt2013} later extended the methods of \cite{Panaretos2010} to compare two populations of non-Gaussian data. 
In a related work, \cite{JARUSKOVA2013} proposed a modification of the test of \cite{Panaretos2010}, used for covariance operator change-point detection. 
\cite{Zhang2015} re-normalized the test statistic of \cite{Panaretos2010} to account for dependence in the data. 
\cite{Gaines2011} proposed a test for equality of two covariance operators based on univariate likelihood ratios and Roy's union intersection principle. 

Up until this point, existing tests were based on the Hilbert-Schmidt metric. 
\cite{Pigoli2014} presented a discussion of distances between covariance operators, including criticisms of using finite dimensional distances on functional data. 
They argued that using a Hilbert-Schmidt metric ignores the geometry of the space of covariance kernels, and therefore is not an appropriate distance. 
As a result, they introduced a two sample permutation procedure, which \cite{cabassi2017} later extended to the multi-sample case. 
In the same vein of resampling, \cite{Paparoditis2016} proposed a $k$-sample bootstrap test that can detect differences in the mean and/or the covariance structure simultaneously.

\cite{Guo2016} further studied a multi-sample test, which was first proposed by \cite{Zhang2013}. 
When the data comes from a Gaussian process, under the null hypothesis, their test statistic is a $\chi^2$-type mixture. 
This distribution must be approximated in practice. 
They also provided a random permutation method to be used in the case of small samples and/or non-Gaussian data. 
\cite{Guo2018} developed a $k$-sample test inspired by functional ANOVA. 
One feature of this test is that it does not require some form of dimension reduction.  
Further, their method is scale invariant in the sense that re-scaling the data at any time $t$ does not affect the test statistic. 
Similar to that of \cite{Guo2016}, the distribution of the test statistic under the null hypothesis must to be estimated. 
The estimation of the critical value relies on parameters estimated from the data, which may pose problems if the data are contaminated. 

\cite{Boente2018} studied a new type of bootstrapping method to calibrate critical values for the covariance kernel testing problem. 
They focused on norms between covariance operators and the resulting distribution under the null hypothesis is based on eigenvalues of fourth moments, which must be estimated. 
They suggested bootstrapping the eigenvalues of fourth moment operators. 
This can be problematic if the data is heavy tailed or contaminated.

\cite{Kashlak2019} provided a concentration inequality based analysis of covariance operators, which includes a $k$-sample test and a classifier. 
They used concentration results to develop confidence sets based on $p$-Schatten norms. 
They then used `tuned' confidence sets to define rejection regions for $k$-sample tests. 
This test tends to underestimate the confidence level in the case where the data is heavy tailed.

Some other related works include the following: 
\cite{Lopez-Pintado2017} used a version of band depth defined on images to test for a difference in dispersion between two sets of images. 
The measure of dispersion ignores shape or `wigglyness' differences between the two samples. 
\cite{Sharipov2019} extended the bootstrapping procedures of \cite{Paparoditis2016} to change-point problems and dependent data. 
\cite{Rice2019} introduced a test for change-points in the cross-covariance operator of two functional time series. 
\cite{Flores2018} presented a test for homogeneity of two distributions based on depth measures. 
They explicitly stated that the paper was not focused on means or covariance operators. 
They provided four test statistics, based on the deepest functions or absolute values of differences in the depth distributions. 
\cite{Aston2017} focused on testing for a condition called separability which is specific to hypersurface data such as f-MRI.

FKWC tests have several advantages over these other methods. 
First, FKWC tests are very robust, a feature that has not often been discussed in other works. 
FKWC tests are based on rank statistics, generated via functional data depth measures. 
Functional data depth measures are, among other things, used for outlier detection and trimmed means; data depth measures are designed to produce robust inference procedures. 
A test statistic based on ranks of data depth measures would thus, inherent the robustness properties of both the depth measure and of rank statistics in general. 
We demonstrate this robustness via simulation in Section \ref{sec::sim} as well as the lack of moment assumptions needed for consistency of the test, see Section \ref{sec::theory}. 

Aside from being robust, many functional data depth measures are invariant under certain transformations of the data \citep{Gijbels2017}. 
If the functional observations are all scaled by an arbitrary function, we would like the test statistic to remain unchanged. 
\cite{Guo2018} points out that many existing tests for equal covariance structure are not invariant under this type of transformation, e.g., those of \citep{Panaretos2010, Guo2016}. 
On the contrary, many data depth measures remain unchanged if the data are scaled by an arbitrary function. 
Such invariance properties are then inherited by the FKWC test statistic, provided derivatives are not included in the calculation of the depth, see Section \ref{sec::dep}. 
If derivatives are included, the FKWC test satisfies a weaker form of transformation invariance. 

Furthermore, using data depth measures allows us to leverage existing consistency results \citep{Nagy2019} and provide asymptotic analysis of the FKWC tests under both the null and alternative hypotheses. 
We show that under the null hypothesis the test statistic is a chi-squared random variable. 
This is a particularly nice feature, as it circumvents the need to estimate the distribution of the test statistic under the null hypothesis using the data, which many other tests require, e.g., \citep{Pigoli2014, Paparoditis2016, cabassi2017, Guo2018, Boente2018, Kashlak2019}. 
Using the data to estimate the distribution of the test statistic under the null hypothesis can be complicated when the data are contaminated. 
This can also be computationally expensive if resampling methods are used. 

Not only is there no need to use the data to estimate the distribution of the FKWC test statistic under the null hypothesis, there is also no need to estimate the sample covariance operators in the computation of the FKWC test statistic. 
This fact implies that one does not need to reduce the dimension of the data via truncated basis expansions of the covariance kernels, as is needed by the methods of \citep{ Fremdt2013, Pigoli2014, Paparoditis2016, Boente2018}. 
An additional byproduct of avoiding estimation of the covariance operators is that we do not require finite fourth moments or any fourth moment related assumptions for our theoretical analysis. 
Such assumptions are required for many other tests, for example, \citep{Panaretos2010, Gaines2011, Fremdt2013, Paparoditis2016, Guo2018, Boente2018} all require some type of fourth moment assumption on the data. 

In terms of the alternative hypothesis, we show that under some mild conditions, the FKWC tests are consistent under a wide class of alternatives. 
We also provide a method for estimating the power and sample size under general alternatives. 
Some recent works have explored various local alternatives for this testing problem \citep{Gaines2011, Guo2016, Guo2018, Boente2018}. 
We also provide a class of local alternatives under which a particular FKWC test is consistent. 
This FKWC test is based on a new depth measure $L^2$-root depth, for which we prove several elementary properties. 
This depth measure has a particular interpretation in this testing problem, which provides the basis for its development. 

The rest of the paper is organised as follows. 
Section \ref{sec::meth} covers the methodology of the hypothesis tests, including the data model and the intuition behind the test statistic. 
Section \ref{sec::dep} gives a brief overview of functional data depths and their benefits when used in a hypothesis testing context. 
Section \ref{sec::theory} presents asymptotic results on the behaviour of the test statistic under both the null and alternative hypotheses. 
Section \ref{sec::sim} presents a simulation study, in which we compare the FKWC tests to some competing tests, including those of \citep{Guo2018, Boente2018}. 
The last section, Section \ref{sec::dataana} shows an application of the FKWC test to two data sets. 
We first compare several samples of intraday financial asset return curves. 
Here, we test to see if the residuals of a functional GARCH model have similar covariance structure. 
We next analyse speech recognition data, where the observations are log periodograms of five groups of recorded syllables. 
We perform FKWC multiple comparisons on these data to determine which pairs of syllables are similar in terms of covariance structure. 
\section{Model and test statistic}\label{sec::meth}
Suppose that we have observed $J$ independent, random samples and that for each sample $j\in\{1,\ldots,J\}$ we have $X_{j1},\dots,X_{jN_j}$ functional observations. 
The combined sample size is then $N=\sumJ N_j,$ where we assume that $N_j/N\rightarrow\vartheta_j$ as $N\rightarrow\infty$. 
We define `functional' observations by the following assumptions: 
First, we assume that each  $X_{ji}$ is a mean square continuous stochastic process, meaning for each $t\in[0,1],$ $X_{ji}(t)$ is measurable with respect to some probability space $(\Omega,\mathscr{A},P)$ and that $\lim_{t\rightarrow s}\E{}{|X(t)-X(s)|^2}=0$. 
Secondly, for each $\omega\in \Omega$,  $X_{ji}(\cdot,\omega)$ is a continuous function. 
We use $\mathfrak{F}$ to denote the space of such processes. 
These assumptions imply that $X_{ji}(t,\omega)$ is jointly measurable with respect to the product $\sigma$-field $\mathscr{B}\times \mathscr{A}$, where $\mathscr{B}$ denotes the Borel sets of $[0,1]$. 
This joint measurability implies that each $X_{ji}$ can be interpreted as a random element which lies in $\ltwo([0,1],\mathscr{B},\mu)$, where $\mu$ is the Lebesgue measure on $[0,1]$. 
For more details see Chapter 7 of \cite{hsing_eubank_2015}. 
We also assume that $\E{}{X_{ji}}=\mathbf{0}$ where $\mathbf{0}$ is the zero function. 
If necessary, in practice, data can be centered by a robust estimator of the mean. 
Some variants of the proposed test involve derivatives, and for those tests we will additionally require that the derivative of $X_{ji}$, by which we denote $X_{ji}^{(1)}(t)$, exists on the interval $(0,1)$ and satisfies the same continuity assumptions imposed on $X_{ji}$. 
We remark that our methods extend to higher dimensional domains and range, i.e., $X_{ji}\colon [0,1]^{d}\rightarrow \re^{p}$ on which functional data depths are defined, but we restrict our study to the univariate domain and range setting. 

The covariance kernel of a mean square continuous stochastic process $X$, whose mean is $\mathbf{0}$, is defined as 
\begin{equation*}
    \kerr(s,t)\coloneqq\E{}{X(t)X(s)}
\end{equation*}
and the associated covariance operator is defined as 
\begin{equation*}
    (\kerrO f)(t)\coloneqq\intd f(s)\kerr(s,t)ds\ .
\end{equation*}
The goal is to construct a test statistic for testing the following hypothesis
\begin{equation*}
    H_0\colon    \kerr_{1}=\dots=\kerr_{J}\qquad \text{v.s.}\qquad H_1\colon  \kerr_{j}\neq\kerr_{k},\ \text{for some }j\neq k.
\end{equation*}
Here $\kerr_j$ refers to the covariance kernel of group $j$. 
The assumptions on the random variables $X_{ji}$ imply that this is equivalent to the hypothesis
\begin{equation*}
    H_0\colon    \kerrO_{1}=\dots=\kerrO_{J} \qquad \text{v.s.}\qquad H_1\colon  \kerrO_{j}\neq\kerrO_{k},\ \text{for some } k\neq j,
\end{equation*}
where $\kerrO_{j}$ is the covariance operator of group $j$.  
We can denote the probability measure over $\ltwo([0,1],\mathscr{B},\mu)$ which describes the random behaviour of $X_{ji}$ by $F_j$. 
Let $$F_*\coloneqq \vartheta_1 F_1+\vartheta_2 F_2+\vartheta_3 F_3+\dots+\vartheta_\ell F_\ell+\vartheta_{J} F_{J},$$
which is a mixture of probability measures over $\ltwo([0,1],\mathscr{B},\mu)$. 
Alternatively, this can be interpreted in the stochastic process sense, such that the finite dimensional distributions of an element from the combined sample $X$: $(X(t_1),\ldots,X(t_k))$ for $k\in \mathbb{N}$, are mixtures of the $J$ finite dimensional distributions of each group, with weights $\{\vartheta_j\}_{j=1}^J$. 
Note that these finite dimensional distributions can be identified by $F_*$ and $F_j$. 

Typically, testing this hypothesis involves estimating $\kerr_j$. 
In order for estimates of $\kerr_j$ to converge weakly, it is typically required that $\E{}{\norm{X_{ji}}^4}<\infty$, in some cases it is not desirable to make this assumption. 
Estimation of $\kerr_j$ is also high dimensional, and can be computationally intensive if repeated a number of times, such as in a bootstrap procedure. 
We take a different approach, and do not aim to estimate $\kerr_j$. 
Instead, the idea is to reduce each observation to a one dimensional rank via a data driven ranking function. 
The ranking function is designed such that differences in the samples' mean ranks are implied by differences in underlying covariance kernels. 
We can then use the classical Kruskal Wallis test statistic \citep{Kruskal1952} and perform a rank test. 
Specifically, the test statistic proposed is
\begin{equation}
   \widehat{\mathcal{W}}_{\scaledN}\coloneqq  \frac{12}{N(N+1)} \sumJ N_j \left(\rhatbarj-\frac{N+1}{2}\right)^2.
    \label{eqn:KW}
\end{equation}
Here, $\rhatbarj$ is the mean rank of the observations in group $j$, where the ranking mechanism will be explained in the next section. 
This test statistic also gives, for each sample $j$, a measure of how much its covariance kernel differs from the average sample kernel via $$\left(\rhatbarj-\frac{N+1}{2}\right)^2.$$ 
Alternatively, we can perform FKWC multiple comparisons, see Section \ref{sec::dataana2}. 

We can further modify this test statistic using the methods of \cite{Gastwirth1965}, who presented a more powerful version of the univariate Wilcoxon Rank-Sum test. 
These were later extended to multivariate, depth-based rank tests \citep{Chenouri2011}. 
We further extend these methods to the functional setting. The percentile modification is predicated on the fact that it is actually the extreme rank values that allow us to detect differences between samples. 
The idea is to remove the middle portion of the data, and only use the outlying data or, equivalently, the low depth-based ranks. 
To this end, let $r\in (0,1)$ and let $N'=\floor{rN}{}$. 
Let $\delta_j(s)=1$ if the observation which has rank equal to $s$ is in group $j$ and let $\delta_j(s)=0$ otherwise. 
Define the percentile modified test statistic as 
\begin{equation}
    \widehat{\mathcal{M}}_{\scaledN,r}\coloneqq \sum_{j=1}^N \left(1-\frac{N_j}{N}\right)K_j \qquad\text{with}\qquad K_j=\frac{1}{\sigma^2_j}\left(\sum_{s=1}^{N'}(N'-s+1)\delta_j(s)-\varrho_j\right)^2,
\end{equation}
where 
\begin{equation*}
    \varrho_j=\frac{N_{j} N^{\prime}\left(N^{\prime}+1\right)}{2 N} \qquad\text{and}\qquad \sigma^2_j=\frac{N_{j}\left(N-N_{j}\right) N^{\prime}\left(N^{\prime}+1\right)\left[2 N\left(2 N^{\prime}+1\right)-3 N^{\prime}\left(N^{\prime}+1\right)\right]}{12 N^{2}(N-1)}.
\end{equation*}
Choosing $r$ is a matter of simulation and will be taken up in Section \ref{sec::sim}. 

\section{Functional depth measures}\label{sec::dep}
The methods we use for ranking the data are based on functional depth measures $\D{\cdot;F}$. 
Functional data depth measures, extended from multivariate data depth measures \citep{Zuo2000}, are versatile, nonparametric tools used in the analysis of functional data. 
Functional data depth measures $\D{\cdot;F}\colon \mathfrak{F}\rightarrow\re$ assign each value in $x\in \mathfrak{F}$ a real number which describes how `central' $x$ is with respect to some measure $F$ over the Hilbert space $\ltwo([0,1],\mathscr{B},\mu)$. 
Here, `central' is used very loosely; central can be in terms of location but also can be in terms of shape. 
Often we have that $F=F_{\scaledN}$, the empirical measure of the data and the depth values $\D{x ;F_{\scaledN}}$ describe centrality with respect to the sample. 
$F_{\scaledN}$ is such that there is $1/N$ mass at each observation in the sample.

Sample ranks based on data depth measures can be calculated simply by ranking the (univariate) depth values of the observations, where depth is calculated with respect to $F_{\scaledN}$. 
In this work, we rank the observations with respect to $F_{*,\scaledN}$, which places equal weight on each element of the combined sample. 
For some $j\in \{1,\ldots,J\}$ and some $i\in\{1,\ldots,N_j\}$, define the sample depth-based rank of $X_{ji}$ to be
$$ \rankji\coloneqq \#\{X_{\ell m }\colon \D{X_{\ell m };F_{*,\scaledN}}\leq \D{X_{ji};F_{*,\scaledN}},\ \ell\in\{1,\ldots,J\},\ m\in\{1,\ldots,N_\ell\}\}\ .$$
This method of ranking is analogous to univariate center outward ranking, in the sense that observations have a high rank when they are deep inside the data cloud.

The motivation for using depth measures as the ranking function follows from the fact that depth has already been shown to have good power for detecting scale changes in surface data \citep{Lopez-Pintado2017}. 
Additionally, using depth ranks for such a purpose has been shown to work well in the univariate and multivariate settings \citep{Sieg1960, ansari1960, Chenouri2011, chenouri2012}. 
In fact, it appears that ranks based on depth measures capture second order differences better in the functional setting than in the multivariate setting. 
For example, a sign change in a covariance parameter could not be detected by data depth ranks in the multivariate setting \citep{Chenouri2011}. 
This is not the case in the functional setting; some of our proposed FKWC tests can detect a difference of the type $\kerr_1(s,t)=-\kerr_2(s,t)$. 
As discussed by \cite{Serfling2017}, functional depth measures are designed with the goals of describing functional data in mind; these goals differ from those of multivariate data description. 
Specifically, the functional data depth measures are not required to be affine invariant and are designed to account for both shape and scale features of the data. 
Differences between the covariance kernels are often exhibited by changes in the shape and/or scale of the data, precisely the features captured by functional data depth functions. 
For example, Figure \ref{fig:PM} shows two samples of 10 Gaussian processes and their derivatives. 
Each sample has the same mean but a different covariance kernel. 
Visually, the distinguishing factor between these two samples is the scale and shape of the curves and their derivatives. 
Notice that the difference is more pronounced in the derivatives. 

\begin{figure}
 \begin{minipage}[c]{0.5\textwidth}
      \includegraphics[width=\textwidth]{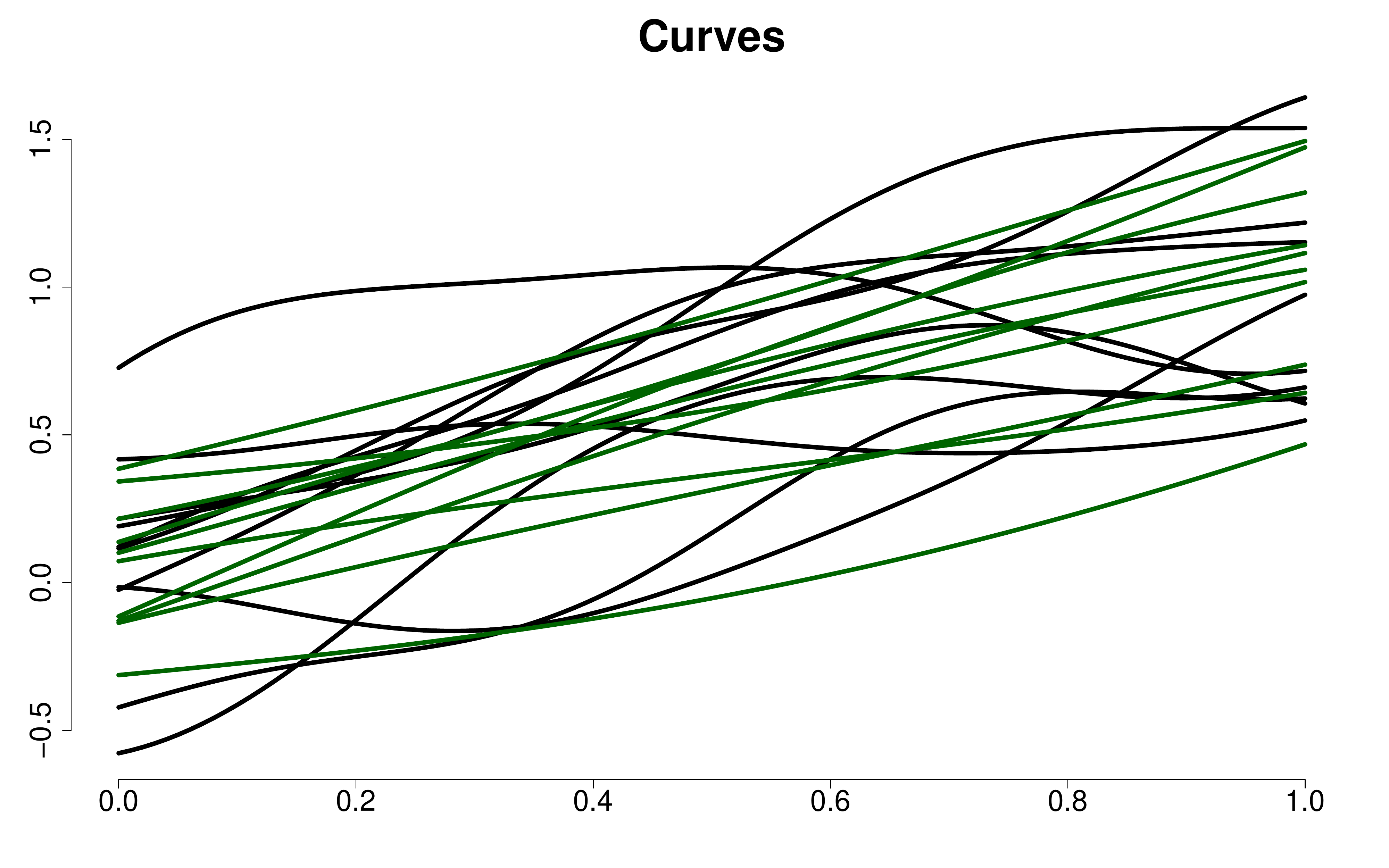}
      \caption*{(a)}
 \end{minipage}
 \begin{minipage}[c]{0.5\textwidth}
      \includegraphics[width=\textwidth]{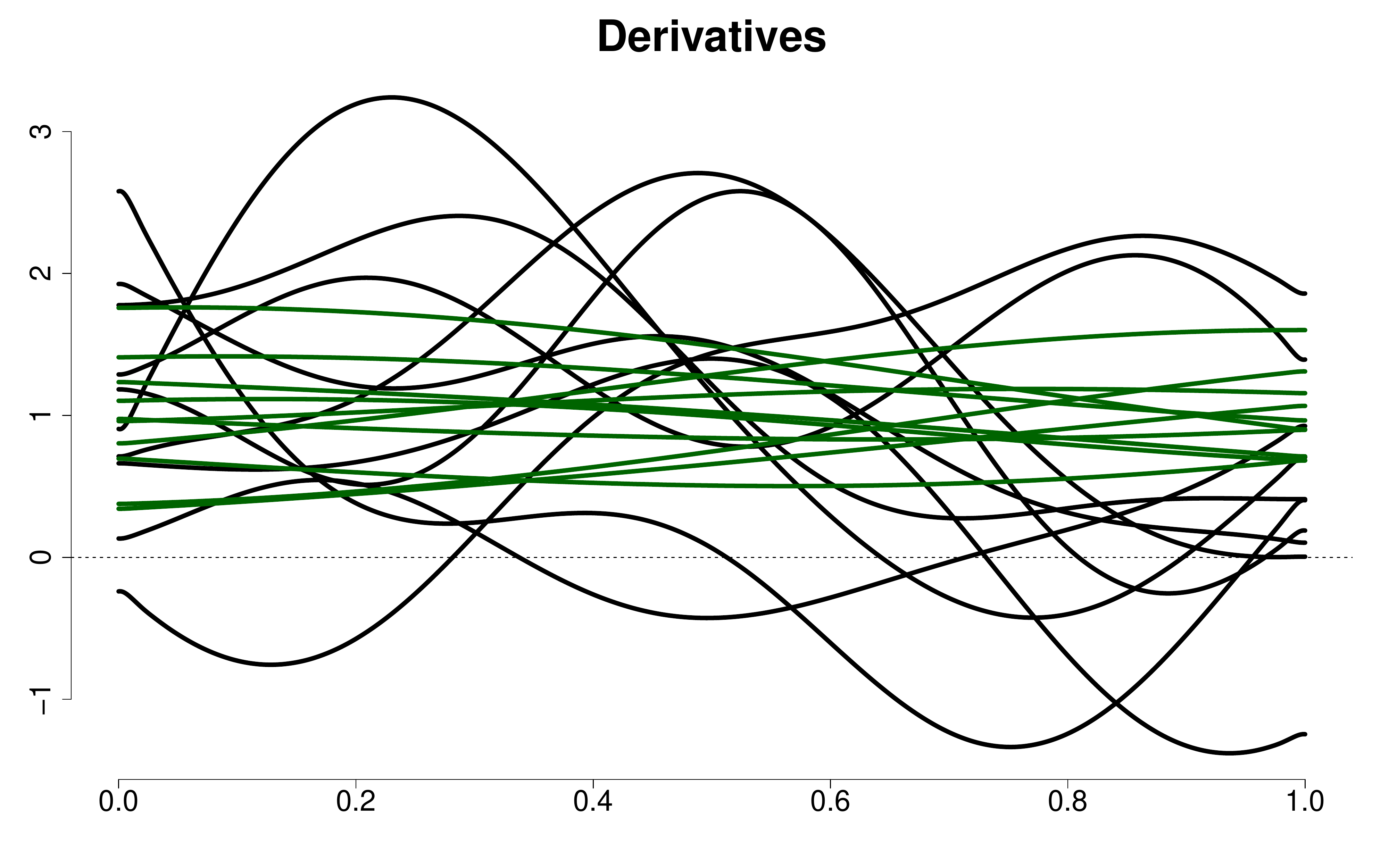}
            \caption*{(b)}
 \end{minipage}
\caption{Two samples of Gaussian processes (a) and their derivatives (b) each with the same mean but a different covariance kernel. The first has an exponential kernel with $\alpha= 0.3$ and $\alpha=1$ in the second sample (see Section \ref{sec::sim}). Notice that the difference in covariance structure is exhibited by a changed in shape of the original curves and a change in shape and scale of the derivatives.}
\label{fig:PM}
\end{figure}

There have been many definitions put forth for functional depth measures, especially recently. 
We restrict ourselves to depth measures that meet criteria suitable for the hypothesis testing problem. 
Namely, we would like the depth measures to admit definitions on multivariate functional data. 
Multivariate functional data refers to observations whose range is $\re^{p}$ for some integer $p>1$. 
For example, the pair of a functional observation and its derivative $(X,X^{(1)})$ could be referred to as a multivariate functional observation with $p=2$. 
Choosing depth measures that admit multivariate functional definitions allows for the use of derivatives, warping functions or other functions in the depth computation. 
In this paper, we use the derivatives in conjunction with the originally observed curves to compute the depth values and show that this leads to better performance.

On top of having multivariate functional definitions, we would like the depth measure to be scale and translation invariant with respect to $F_{\scaledN}$. 
By scale invariant, we mean that if the data are collectively scaled by some non-zero function $a(t)$ we want the results of the test to remain unchanged. 
Clearly translation invariance is necessary. 
Many functional depth measures are invariant under the aforementioned transformations of the data. 
However, suppose that we have a set of univariate functional observations but we take the data to be the observed function along with the observed first derivative $(X,X^{(1)})$. 
Scaling each $X$ by $a$ results in a different transformation of the derivative and produces different observations
$$(aX,aX^{(1)}+a^{(1)}X).$$
If $a$ is a constant function, then a scaling of the data corresponds to a scaling of the derivative and all is well. 
Therefore, ranking methods that are based on analysis of the function and its derivative are not generally invariant under scaling by a non-constant function. 
They are, however, invariant under scaling by a constant function. 
In other words, the FKWC test based on the original curves is invariant under scaling by a non-constant function whereas the test based on the paired observations of functions and their derivatives is not invariant under scaling by a non-constant function. 
However, in order for $X^{(1)}$ to be affected overly by arbitrary scaling, the derivatives of $a$ must be large relative to function $X$, suggesting $a$ is quite steep in some places. 
It is hard to see a reason to scale the data by steep functions in practice.

It is also desirable for the depth measures to be uniformly consistent with a rate of $O(N^{-1/2})$ under some general conditions, which gives that for large $N$, the test will detect a difference in covariance structure with high probability. 
So far, this has only been shown for the depth of \cite{Fraiman2001} by \cite{Nagy2019}. 
We extend their result to the random projection depth in Section \ref{sec::theory}. 
Lastly, we use functional data depth measures which have implementations in \texttt{R} or software that can be integrated into \texttt{R}. 
This last restriction allows us to readily provide simulation results for the tests, as well as implementations that can be used in practice.

The first depth measure we discuss is the integrated depth of \cite{Fraiman2001},
Let $F_t$ be the finite dimensional distribution of $X(t) \in\re^p$ if $X\sim F$. 
For each fixed $t$, we can define the depth of $x(t)$ with respect to $F_t$ as $\D{x(t);F_t}$ where $\D{\cdot}$ is a univariate or multivariate depth function. 
It is natural to define the functional depth as 
$$\intd \D{x(t);F_t} dt;$$
the average of these pointwise depths.  
In this work, we let $\D{\cdot}$ be the multivariate halfspace depth \citep{Tukey1974}, which is defined as
$$\HD(x(t);F_t)=\inf_{y\in \rd{p},\ \norm{y}=1} \Pr(\ip{X(t), y}\leq \ip{x(t), y}).$$
With $\Dd=\HD$ the depth of \cite{Fraiman2001} is equivalent to a version of 
the multivariate functional halfspace depth \citep{slaets11, Hubert2012, Claeskens2014}. 
Therefore, for the remainder of the paper we shall refer to this depth as the multivariate functional halfspace depth:
\begin{equation*}
\MFHD(x ;F)=\intd \HD(x(t);F_t)dt.
\end{equation*}
Multivariate functional halfspace depth is invariant under transformations of the type $T(X)=AX+b$ where for each $t$, $A(t)$ is an invertible matrix and $b$ is a smooth function. 
Note the remark above; methods using derivatives are only invariant under constant scaling. 
\cite{Nagy2019} have recently shown that the sample $\MFHD$ depths converge uniformly to their population counterpart at a rate of $N^{-1/2}$. 
We can use the \texttt{MFHD} \texttt{R} package to compute this depth or the \texttt{fda.usc} package via \texttt{depth.FMp}. 
We opt for the latter in simulation in Section \ref{sec::sim}. 
Further, we use $\MFHD'(x;F)$ to refer to $\MFHD$ applied to the pair $(x,x^{(1)})$. 

We next consider modified band depth \citep{Lopez-Pintado2009}. 
For some $x\in \mathfrak{F}$ and a set of functional observations $Y_1,\dots,Y_k$ we can define
$$B\left(x;Y_1,\dots,Y_{k}\right)=\left\{t: t \in [0,1], \min _{r=1, \ldots, k} Y_{i_r}(t) \leq x(t) \leq \max _{r=1, \ldots, k} Y_{i_r}(t)\right\},$$
as the set such that $x$ is in the $k$-band delimited by $Y_1,\dots,Y_k$. If $Y_1,\dots,Y_k$ are independent and come from the distribution $F$, we can define
\begin{equation*}
    \MBD_k^{(x)}=\E{F}{\mu(B\left(x;Y_1,\dots,Y_{k}\right))}
\end{equation*}
where $\mu$ is the standard Lebesgue measure. 
Then the `modified' band depth with parameter $K$ is equal to 
$$\MBD_{\scaled{K}}(x;F)=\sum_{k=2}^K \MBD_k^{(x)}.$$ 
This depth measure is invariant under the linear transformations described previously \citep[see property ``P1-F'' in][]{Gijbels2017}.
There are two multivariate functional extensions of this depth measure, \citep{Ieva2013} and \citep{LP2014}. 
We choose the extension of \cite{Ieva2013} because of its existing implementation in \texttt{R}. 
Essentially, the multivariate extension of the modified band depth as in \cite{Ieva2013} is defined as 
$$\MMBD_{\scaled{J}}(x;F)\coloneqq\sum_{k=1}^{p} w_k\MBD_{\scaled{J}}(x;F^{k}), \qquad\text{and}\qquad \sum_{p=1}^{p}w_p=1,$$
where $F^k$ is the marginal distribution pertaining to the $k^{th}$ univariate functional argument in the vector of observations. 
This depth is also invariant under the discussed transformations. 
We will use $\MBD'_K(x;F)$ to refer to $\MMBD_{\scaled{K}}$ applied to the pair $(x,x^{(1)})$.

Functional spatial depth, described by \cite{Chakraborty2014}, is the extension of multivariate spatial depth to the functional setting. 
Define spatial depth to be $$\SD(x;F)\coloneqq 1-\norm{\E{F}{s(x-X)}},$$ 
where 
$$s(y)=\left\{\begin{array}{cc}
    y/\norm{y} & \norm{y}\neq 0 \\
     0 & o.w.
\end{array}\right. .$$
Here, $\norm{\cdot}$ refers to the $\ltwo([0,1],\mathscr{B},\mu)$ norm. 
We can also define the kernelized spatial depth \citep{Sguera2014}. 
Consider a mapping $\phi$ from $\ltwo([0,1],\mathscr{B},\mu)$ into some feature space and a corresponding kernel function, different from the covariance kernels described above, $\gamma\colon \ltwo([0,1],\mathscr{B},\mu)\times\ltwo([0,1],\mathscr{B},\mu)\rightarrow\re$ as 
$$\gamma(x,z)=\ip{{\phi(x)},{\phi(z)}}.$$
We can then define the kernel depth as 
$$\KSD(x;F)=1-\norm{\E{F}{s(\phi(x)-\phi(X))}}.$$
These depth measures are transformation invariant as above if $a$ is surjective, which is a fairly mild restriction.
We can extend these depth measures to the multivariate functional setting by computing the spatial functional depth values marginally and then taking an equally weighted average of such marginal depth values; analogous to what is done above for the modified band depth \citep{Ieva2013}. 
This modification will be referred to as $\SD'$ for spatial depth and $\KSD'$ for kernelized spatial depth. 
One might also take a multivariate depth of the marginal functional depths, but this is more computationally costly and we think unnecessary in this context.
Both kernelized and standard spatial functional depths are implemented in the \texttt{fda.usc} \texttt{R} package. 

The last of the existing depth measures we introduce is the random projection depth \citep{Cuevas2007}. 
The idea behind this depth measure is to choose $M$ random directions $u_m(t)\in S$ where $S$ is the unit sphere in $\ltwo([0,1],\mathscr{B},\mu)$. 
Then, for each direction $u$, compute a separate depth value based on the projections onto $u$, i.e., $\D{\ip{x,u};F_u}$. 
Here, $F_u$ is the CDF of the random variable $\ip{X,u},\ X\sim F$. 
These depths $\D{\ip{x,u};F_u}$ are then averaged to give a final resulting depth value, viz.
$$\RP_{\scaled{M}}(x;F)=\frac{1}{M}\sum_{m=1}^M\D{\ip{x,u_m};F_{u_m}}.$$ 
In this work, we typically take $\D{\ip{x,u};F_u}=F_u(x)(1-F_u(x))$ and use $M=20$ projections. 
\cite{Cuevas2007} introduce a second version of the random projection depth in which they calculate both the projection of the function and the projection of the function's first derivative, which provides pairs of observations. 
They then use a multivariate depth on the couples:
$$\RP_{\scaled{M}}'(x;F)=\frac{1}{M}\sum_{m=1}^M\D{(\ip{x,u_m},\ip{x^{(1)},u_m});F_{u_m,(1)}},$$
where $F_{u,(1)}$ is the bivariate distribution of $(\ip{X,u},\ip{X^{(1)},u})$. 
For $\D{\cdot}$ we use the likelihood depth \citep{MULLER2005153} as this is the default in the \texttt{fda.usc} \texttt{R} package. 
Note that this depth does not have any consistency or invariance properties presented in the original paper. 
We discuss consistency in Section \ref{sec::theory} under the assumption that the unit vectors are drawn from a compact subset of $S$. 
It is necessary to restrict to some compact set, seeing as the uniform measure on the unit sphere in infinite dimensional spaces does not exist. 
The choice of directions is therefore an interesting point of discussion. 
For example, instead of selecting the directions from a general subset of $S$, one can take a semi-parametric approach and assume that $X_{ji}=c_{ji}^\top \mathbf{Y}$ where $\mathbf{Y}=(y_1,\ldots,y_B)^\top$ is a vector of orthogonal basis functions. 
It is then natural to use $\mathbf{Y}$ as the set of `directions' for $\RP$ depth. 
Another idea would be to use a robust version of the estimated principal components. 
We leave the study of these methods for future work.

Lastly, we study a new depth related to spatial depth, which we call $L^2$-root depth, or $\LTR$ depth for short. 
We can define $\LTR$ depth as
\begin{equation}
    \LTR(x;F)=\left(1+\E{F}{\norm{x-X}^2}^{1/2}\right)^{-1},
\end{equation}
which is inspired by $L^p$ depths of \cite{Zuo2000}. 
The following theorem lists some properties of $\LTR$ depth. 
\begin{theorem}
Let $X\sim F$, where $F$ is a measure over $\mathscr{L}^2([0,1],\mathscr{B},\mu)$, $F_{\scaledN}$ be the empirical measure corresponding to a random sample of size $N$ from $F$, $a,b\in \mathfrak{F}$ and $c,c'\in \re^+$, then $\LTR$ depth satisfies the following properties
\begin{enumerate}
    \item Sample ranks based on $\LTR$ depth are invariant when the sample is transformed by a linear function $h$, such that $h(x)=a\,x+b$. 
    \item If $X\eqd -X$, then $\sup_x\LTR(x;F)=\LTR(\mathbf{0};F)$. 
    \item If $X\eqd -X$, then $\LTR(c\,x;F)$, is decreasing in $c$. 
    \item $\lim\limits_{c\rightarrow\infty} \LTR(c\,x;F)=0$.
    \item Suppose that $\E{}{\norm{X}^2}<\infty$, then $\sup\limits_x|\LTR(x;F_{\scaledN})-\LTR(x;F)|=o(1)\ a.s.\ $. 
\end{enumerate}
\label{thm::ltd_prop}
\end{theorem}
\noindent The proof of this theorem is given in Appendix \ref{sec::app}. 
We remark that $\LTR(x;F)$ is not invariant under linear transformations as given by $h$ in property 1.\ of Theorem \ref{thm::ltd_prop}. 
It is easy to see that if
$$\LTR(x;F)=\frac{1}{1+c'}\qquad\text{then}\qquad\LTR(ax;aF)=\frac{1}{1+\norm{a}c'},$$
where $aX\sim aF$, with $X\sim F$. 
This fact implies that hypothesis tests based on $\LTR$ depth values themselves won't be invariant under linear transformations. 
However, Theorem \ref{thm::ltd_prop} shows that hypothesis tests based on ranks of these depth values invariant under linear transformations. 
Further, a median based on this depth would be equivariant under linear transformations as given by $h$ in property 1.\ . 

Motivation for the study of $\LTR$ depth in this context comes from the relationship between the covariance operator of a random variable and its expected, squared $L^2$ norm.
In this setting, we see that ranking the observations based on $\LTR$-depth is equivalent to ranking the observations based on their norms. 
Under the assumption of zero mean, we have that
$$\E{F_*}{\norm{X_{ji}-X}^2}=\E{F_*}{\norm{X_{ji}}^2}+\E{F_*}{\norm{X}^2}=\norm{X_{ji}}^2+\sum_{j=1}^J \vartheta_j \mathscr{K}_j+o(1),$$
from which it is easily seen that the ranks are equivalent to those based on $\E{F_*}{\norm{X_{ji}}^2}$. 
We emphasize that this relationship relies on the assumption of zero mean, and the data must be centred. 
In this context, ranks generated from this depth function do not need to be estimated; we can compute ranks based on $\D{\cdot;F_*}$ directly. 
It should be noted that the aspects of the distribution captured by $\LTR$ depth are limited to the magnitude of the observations. 
To account for this, we propose incorporating shape information, which can be obtained from the derivatives of the observations. 
This depth can be extended to account for $p$ derivatives with
$$\LTR_p(x;F)=\left(1+\frac{1}{p}\sum_{k=0}^{p}\E{F}{\norm{x^{(k)}-X^{(k)}}^2}^{1/2}\right)^{-1}.$$


\section{Theoretical results}\label{sec::theory}





This section is devoted to characterizing the behaviour of the FKWC tests under the null and alternative hypotheses when the sample size is large. 
Note all proofs of theorems presented in this section can be found in Appendix \ref{sec::app}. 
\begin{theorem}
Suppose that $\Pr(\widehat{R}_{ji}=\widehat{R}_{k\ell})=0$ for all $ji\neq k\ell$. 
Then, under the null hypothesis 
$$\widehat{\mathcal{W}}_{\scaledN}\cond \rchi^2_{\scaled{J}-1}\text{ as } N\rightarrow\infty\qquad \text{and}\qquad \widehat{\mathcal{M}}_{\scaledN,r}\cond\rchi^2_{\scaled{J}-1} \text{ as } N\rightarrow\infty.$$
\label{thm:::null}
\end{theorem}
Tied ranks can be randomly broken to meet the requirements of Theorem \ref{thm:::null}. 
The behaviour under the null hypothesis is remarkably simple for such a complex testing problem. 
Therefore, the critical values can easily be obtained independently of the data and thus, this aids accuracy and computation time. 
The fact that the critical values are independent of the data is important for robustness, as there is no need to assess the robustness of procedures used to approximate the null distribution. 
Under the alternative hypothesis, we must impose additional assumptions in order to have consistency of the test. 
\begin{ass} For all $j$, it holds that $\Pr(\D{X_{j1};F_{*}}\leq v)$, as a function of $v$, is a Lipschitz function.  
\label{ass:cont}
\end{ass}
\begin{ass} It holds that
$\E{}{\sup_{x\in \mathfrak{F}}|\D{x;F_{N,*}}-\D{x;F_*}|}=O(N^{-1/2}).$
\label{ass:consDepth}
\end{ass}
\begin{ass}
There exists $k \in \{1,\ldots,J\}$ such that
\begin{equation}
      \sum_{j=1}^J\vartheta_j \Pr\left(\D{X_{k 1};F_*}>\D{X_{j1};F_*}\right)\neq \frac{1}{2}.
    \label{eqn:assu}
\end{equation}
\label{ass:med_d}
\end{ass}
Assumption \ref{ass:cont} is generally satisfied when $F_*$ is continuous. 
Assumption \ref{ass:consDepth} has been shown to be satisfied for $\MFHD$, see \citep{Nagy2019}.  
We can extend the results of \cite{Nagy2019} to $\RP$ depth. 
Suppose that the unit vectors are drawn from some compact set $\mathcal{C}$ and that $M_N=O(N)$. 
Let 
\begin{equation}
\RP_\infty(x;F)=\int_\mathcal{C}F_u(\ip{x,u})(1-F_u(\ip{x,u})) d\nu(u),
    \label{eqn::rpinf}
\end{equation}
where $\nu$ is the uniform measure on $\mathcal{C}$. 
Then, 
{\small
\begin{align*}
    \E{}{\sup_{x\in\mathfrak{F}}|\RP_{M_N}(x;F_N)-\RP_\infty(x;F)|}&=\E{}{\bigg|M_N^{-1}\sum_{m=1}^{M_N}\D{\ip{x,u_m};F_{N,u_m}}-\int_\mathcal{C}\D{\ip{x,u};F_{u}}d\nu(u)\Bigg|}\\
    &\leq  \E{}{M_N^{-1}\sum_{m=1}^{M_N}|\D{\ip{x,u_m};F_{N,u_m}}-\D{\ip{x,u_m};F_{u_m}}|}+O(N^{-1/2})\\
    &\leq \E{}{4M_N^{-1}\sum_{m=1}^{M_N}  \E{}{\sup_{z\in\re}|F_{N,u_m}(z)-F_{u_m}(z)|\Bigg|u_1,\ldots u_{_{M_N}}}}+O(N^{-1/2})\\
    &=O(N^{-1/2}),
\end{align*}}
where the first inequality is from the triangle inequality and Hoeffding's inequality, and the last equality results from the Dvoretzky–Kiefer–Wolfowitz inequality. 
Note that the same analysis applies when $F_u(\ip{x,u})(1-F_u(\ip{x,u}))$ is replaced with $1/2-|1-F_u(\ip{x,u})|$ in \eqref{eqn::rpinf}. 
Assumption \ref{ass:consDepth} need not be satisfied for LTR depth since the sample ranks are already based on $\LTR(\cdot;F)$. 
To our knowledge, there does not exist results on the rates of convergence for (kernelized) spatial depth or modified band depth. 
\begin{theorem}
If Assumptions  \ref{ass:cont}-\ref{ass:med_d} hold, then for any $\delta>0$, it follows that 
$$\Pr\left(\widehat{\mathcal{W}}_{\scaledN}>\delta\right)\rightarrow 1, \text{ as } N\rightarrow\infty.$$
\label{thm:::alt}
\end{theorem}
Theorem \ref{thm:::alt} is quite general in terms of the alternatives under which the test is consistent. 
The set of alternative hypotheses which induce consistency are contained completely in Assumption \ref{ass:med_d}. 
First, Assumption \ref{ass:med_d} implies that if there is a difference in covariance operator between groups, the difference is exhibited in the depth values (by at least one pair of groups, and not all).
Specifically, the median of $\D{X_{\ell1};F_*}-\D{X_{j1};F_*}$ is non-zero for some $\ell\neq j$. 
Assumption \ref{ass:med_d} has an additional requirement that says the differences between the groups do not perfectly `cancel' each other out. 
This is not an issue if $J=2$. 

Assumption \ref{ass:med_d} does not appear to be directly related to the covariance operator; as a result it is beneficial to demonstrate that changes in the covariance kernels produce, on average, a location difference in the depth values. 
As an example, we show this for the $L^2$-root depth and the random projection depth measures when $J=2$. 
The aim is to say that if two samples differ in covariance operator, then their depth values differ in location. 
Let the sample rank of $X_{ji}$ based on the true distribution $F_*$ be defined as
$$R_{j i}:=\#\left\{X_{\ell m}: \mathrm{D}\left(X_{\ell m} ; F_{*}\right) \leq \mathrm{D}\left(X_{j i} ; F_{*}\right),\ \ell\in\{1,\ldots,J\},\ m\in\{1,\ldots,N_j\}\right\}.$$
Then we define $\mathcal{W}_{\scaledN}$ as the test statistic based on these ranks, which are unknown except in the special case of the $L^2$-root depth-based ranks. 

\begin{theorem}[$L^2$-root Depth]
Suppose Assumption \ref{ass:cont} holds, $J=2$, and that
\begin{equation}
\E{} {\norm{X_{11}}^2-\norm{X_{21}}^2}\neq 0 \text{ implies } \med\left(\norm{X_{21}}^2- \norm{X_{11}}^2\right)\neq 0.
\label{eqn::sdc}
\end{equation}
Then the test based on $\mathcal{W}_{\scaledN}$, using ranks based on $\LTR$ depth, is consistent in the sense of Theorem \ref{thm:::alt} under alternatives of the form $H_1\colon \norm{\kerrO_1}_{TR}\neq\norm{\kerrO_2}_{TR},$ where $\norm{\cdot}_{TR}$ refers to the trace norm. 

\label{thm::SD}
\end{theorem}
Note that, by assumption, $\kerrO_j$ are trace class since the observed processes are mean square continuous; the kernel is continuous. 
Theorem \ref{thm::SD} shows that the mean depth values differ if the trace norm of the covariance operators differ. 
This alternative hypothesis is equivalent to $\sum_{k=1}^\infty\lambda_{k,1}\neq\sum_{k=1}^\infty\lambda_{k,2},$
where $\{\lambda_{k,1}\}_{k=1}^\infty$ and $ \{\lambda_{k,2}\}_{k=1}^\infty$ are the decreasing sequences of eigenvalues resulting from the singular value decomposition of $\kerr_1$ and $\kerr_2$, respectively. 
Clearly if the covariance operators are equal then $\norm{\kerrO_1}_{TR}=\norm{\kerrO_2}_{TR}$ and \eqref{eqn::sdc} is not satisfied.

It is useful to mention a few cases where \eqref{eqn::sdc} is surely satisfied. 
If, for all $j$, $\norm{X_{ji}}^2$ has a symmetric distribution then \eqref{eqn::sdc} is satisfied. 
If $$\norm{\kerrO_1}_{TR}-\norm{\kerrO_2}_{TR}\geq \Var{\norm{X_{11}}^2}^{\frac{1}{2}}+\Var{\norm{X_{21}}^2}^{\frac{1}{2}}$$
then \eqref{eqn::sdc} is satisfied. 
If the distribution of the squared norms is unimodal then we only need the size of the change to satisfy 
$$\norm{\kerrO_1}_{TR}-\norm{\kerrO_2}_{TR}\geq \left(\frac{3}{5}\left(\Var{\norm{X_{11}}^2}+\Var{\norm{X_{21}}^2}\right)\right)^{\frac{1}{2}},$$
in order for \eqref{eqn::sdc} to be satisfied. 
If $X_{ji}$ are Gaussian processes $\mathcal{GP}(0,\kerrO_j)$, then we have that 
$$\norm{X_{11}}^2=\sum_{k=1}^\infty \lambda_{k,1} V_k \qquad \text{and} \qquad \norm{X_{21}}^2=\sum_{k=1}^\infty \lambda_{k,2} V_k',\ \text{ where } V_k,V_k'\overset{iid}{\sim} \chi^2_1.$$
It follows that the random variables $\norm{X_{11}}^2,\ \norm{X_{21}}^2$ are stochastically ordered, implying that \eqref{eqn::sdc} is satisfied. 
The next theorem gives a similar result for the random projection depth. 
\begin{theorem}[Random Projection Depth]
If $X\sim F_*$, let $F_{*,u}$ be the distribution of $\ip{X,u}$. 
Suppose that $u_1,\ldots,u_{_{M_N}}$ are drawn uniformly from some compact set $\mathcal{C}$ and that $M_N=O(N)$. 
Assume that for any $u$, $F_{*,u}$ is three times differentiable, and that the first three derivatives of $F_{*,u}$ are bounded functions in $u$. 
Suppose that Assumption \ref{ass:cont} holds for $\RP_\infty$, $J=2$, $\E{}{\norm{X_{ji}}^3}<\infty$,
and that 
\begin{equation}
\E{} {\RP_\infty(X_{11};F_{*})-\RP_\infty(X_{21};F_{*})}\neq 0 \implies  \med{\left(\RP_\infty(X_{11};F_{*})-\RP_\infty(X_{21};F_{*})\right)}\neq 0.
\label{eqn::rpdc}
\end{equation}
Then FKWC test based on $\widehat{\mathcal{W}}_{\scaledN}$ using $\RP_{M_N}$ with $\D{z;F}=F(z)(1-F(z))$ is consistent in the sense of Theorem \ref{thm:::alt} under alternatives of the form  $$H_1\colon\int_\mathcal{C}\mathcal{H}(F_{*,u})\ip{\kerrO_1 u,u}d\nu(u) +\mathcal{R}_1\neq\int_\mathcal{C}\mathcal{H}(F_{*,u})\ip{\kerrO_2 u,u}d\nu(u)+\mathcal{R}_2,$$
where $\mathcal{H}$ is a function of a univariate distribution function such that $\mathcal{H}(F)\coloneqq \frac{1}{2}f'(0)-(F(0)f'(0)-f^2(0)),$
and the remainder term $\mathcal{R}_j$ can be written as
$$\mathcal{R}_j=\frac{1}{6}\int_\mathcal{C}\E{}{\bigintssss_{0}^{\ip{X_{j1},u}} \left(f_{u, z}^{(2)}(t)(1-2F_{u, *}(t))- 6 f_{u, *}(t) f_{u, *}^{(1)}(t)\right)\left(\ip{X_{j1},u}-t\right)^{3} d t}d\nu(u).$$
\label{thm::RPD}
\end{theorem}

Note that we expect $\mathcal{R}_j$ to be small because $X_{j1}$ have zero mean. 
In terms of $\mathcal{H}(F_{u,*})$, as $F_{u,*}$ approaches symmetry, $\mathcal{H}(F_{u,*})$ approaches $f_{u,*}^2(0)$, the squared height of the projected density at 0. 
Symmetry of the projected distributions implies a form of symmetry on $F_*$. 
This is a natural definition of symmetry, in the sense that it is analogous to halfspace symmetry in the multivariate setting. 
This discussion implies that this method will work better for $F_*$ that are symmetric. 
Another benefit of using projections is the fact that projections characterize the the distribution $F_{*}$ \citep[see][Chapter 7]{hsing_eubank_2015}. 
Therefore, provided we have ``enough'' directions, the test based on $\RP$ will be able to detect arbitrary differences in covariance kernels. 
Theorem \ref{thm::RPD} therefore, shows that all depth functions may not be equivalent in terms of which types of alternatives they can detect. 


To give more insight into the behaviour of the test statistic, we provide some analysis under local alternatives. 
Suppose Assumptions \ref{ass:consDepth} and \ref{ass:med_d} hold. 
Following \cite{Fan2011}, $\mathcal{W}_{\scaledN}$ is approximately distributed as a non-central chi-squared random variable $\rchi^2_{\scaled{J}-1}(\tau_{\scaledN})$ with non-centrality parameter
\begin{equation}
    \tau_{\scaledN}=\frac{12}{N(N+1)} \sum_{j=1}^{J} N_{j}\left\{N \sum_{k \neq j} \vartheta_{k}\left(\Pr\left(\D{X_{j 1};F_*} \leq \D{X_{k 1};F_*}\right)-\frac{1}{2}\right)\right\}^{2}.
    \label{eqn::tau}
\end{equation}
For $L^2$-root depth, we are able to compute $\Pr(\D{X_{k 1};F_*} \leq \D{X_{j 1};F_*}),$ which is the same quantity as $\Pr(\norm{X_{k1}} \leq \norm{X_{j1}}).$ 
Therefore one can compute the power and consequently sample sizes for any assumed $F_*$, using
$$\Pr\left(\sum_{m=1}^p\left[\norm{X_{k1}^{(m)}}-\norm{X^{(m)}_{j1}}\right] \leq 0\right).$$
This could of course be done by Monte Carlo simulation for complicated models.   
\begin{theorem}[Local Alternative Analysis-$L^2$-root depth]
Suppose that for all $i,j,k,\ell$ 
\begin{equation}
\norm{X_{ji}}^2\eqd \left[\frac{\sqrt{N}+\delta_k}{\sqrt{N}+\delta_j}\right]
\norm{X_{k\ell}}^2 \sim G,
\label{eqn::alt-spat2}
\end{equation}
where $G$ has a continuously differentiable density $g$ for some real-valued $\delta_j$s. 
Let $\overline{\delta}=\sum_{j=1}^J \vartheta_j\delta_j$ then, under the $\LTR$ depth ranks, $$\mathcal{W}_{\scaledN}\cond\rchi^2_{\scaled{J}-1}(\tau)\text{ with }\tau=12 \left(\int_{\mathbb{R}} z g(z)^{2} d z\right)^{2} \sum_{j=1}^{J} \vartheta_{j}\left(\delta_{j}-\overline{\delta}\right)^{2}.$$
 \label{thm::sdla}
\end{theorem}
Note that \eqref{eqn::alt-spat2} holds when $ \kerrO_j= \kerrO_0\left[1+N^{-1/2} \delta_j\right]$ and $\norm{X_{ji}}^2$ form a scale family. For example, if $X_{ji}\sim \mathcal{GP}(0,\kerrO_0\left[1+N^{-1/2}\delta_j\right])$, then Theorem \ref{thm::sdla} is applicable. 


Due to the fact that Euclidean spaces are Hilbert spaces, the previous results provide some consequences for similar methods based on depth-based ranks in the multivariate setting. 
For example, \citep{Chenouri2020DD,Ramsay2019b} provide methodologies for detecting single and multiple change-points in the covariance matrix of multivariate data, based on data depth ranks.
Theorem \ref{thm::SD} provides justification for assuming the hypothesis of Theorem 2 of \citep{Chenouri2020DD} as well as for Assumption 4 of \citep{Ramsay2019b}, when the ranks are based on $L^2$-root depth. 
Note that the definition of spatial depth in \citep{Chenouri2020DD} provides equivalent ranks to those of $L^2$-root depth. 
Theorem \ref{thm::SD} then implies that the methods of \citep{Chenouri2020DD,Ramsay2019b} can detect changes in the sum of the eigenvalues of the covariance matrix. 
Similarly, Theorem \ref{thm::RPD} provides justification for assuming the hypothesis of Theorem 2 of \citep{Chenouri2019} and Assumption 4 of \citep{Ramsay2019b} under the multivariate depth studied by \cite{Cuevas2009}. 
\cite{Liu2006} provides a $k$-sample test for the covariance matrix of multivariate data. 
Theorems \ref{thm::SD}, \ref{thm::RPD} and \ref{thm::sdla} give analogous results for this multivariate $k$-sample test. 

\section{Simulation results}\label{sec::sim}
%
In this section we evaluate the finite sample performance of the FKWC tests. 
We compare the performance of the tests using the different depth functions discussed in Section \ref{sec::dep}, as well as the effect of the percentile modification discussed in Section \ref{sec::meth}. 
We further compare the FKWC test against seven other tests: the test of \cite{Boente2018} $Boen$, the $L^2$-norm tests \citep{Guo2016} $Tmax,\ L2nv,\ L2br,\ L2rp$ and the ANOVA inspired tests of \cite{Guo2018} $GPFnv,\ GPFrp,\  Fmax$. 
For the test of \cite{Boente2018} we used 10 principle components and 5000 bootstrap samples. 
For the tests of \cite{Guo2016, Guo2018} that required permutations, we used 1000 permutations. 

We simulated data from both infinite and finite dimensional models, we start with the infinite dimensional models. 
For the infinite dimensional models, we simulated observations from $J=2$ and $J=3$ samples. 
Sample sizes of $N_1=N_2=50,\ 100 $ and 250 for the two sample case were used. 
For the three sample, we used $N_1=N_2=N_3=50,$ and 100. 
Since the results of the three sample case were so similar to those of the two sample case, we only report the results of the two sample case. 
In each infinite dimensional case, the data were sampled from either a Gaussian process $\mathcal{GP}$, a student-t process with one degree of freedom $t_1$, or a skewed Gaussian process $\mathcal{SG}$. 
For the infinite dimensional runs we used a squared exponential covariance kernel 
$$\kerr_{}(s,t; \alpha,\beta)=\beta\ \exp\left(\frac{-(s-t)^2}{2\alpha^2}\right),$$
and the sample differences were controlled via $\alpha$ and $\beta$.

\begin{figure}[t]
\begin{minipage}[c]{.49\textwidth} 
    \centering
    \includegraphics[width=0.9\textwidth]{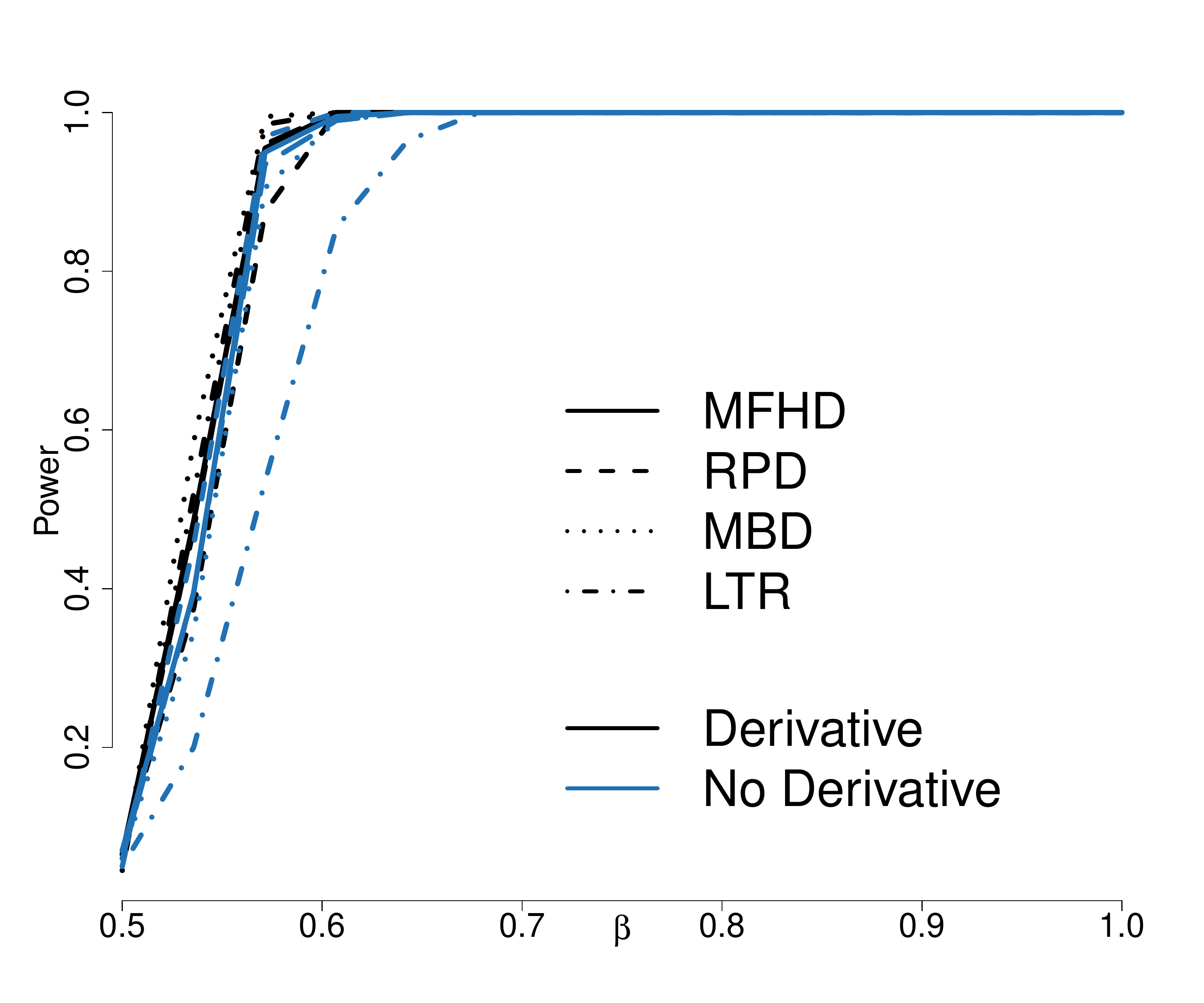}
    \caption*{(a)}
\end{minipage}
\begin{minipage}[c]{.49\textwidth} 
    \centering
    \includegraphics[width=0.9\textwidth]{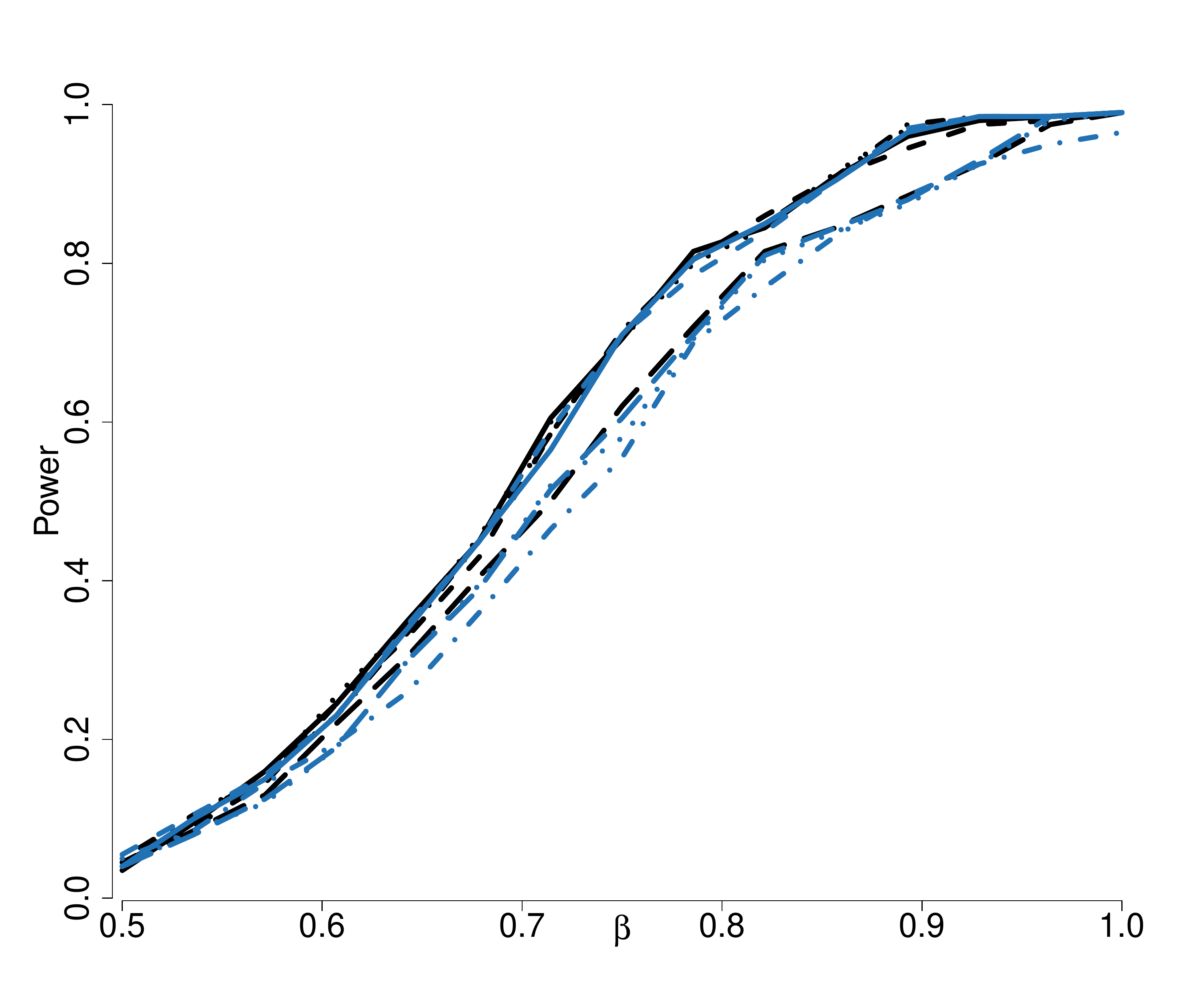}
    \caption*{(b)}
\end{minipage}
\hfill\newline
\begin{minipage}[c]{.49\textwidth} 
    \centering
    \includegraphics[width=0.9\textwidth]{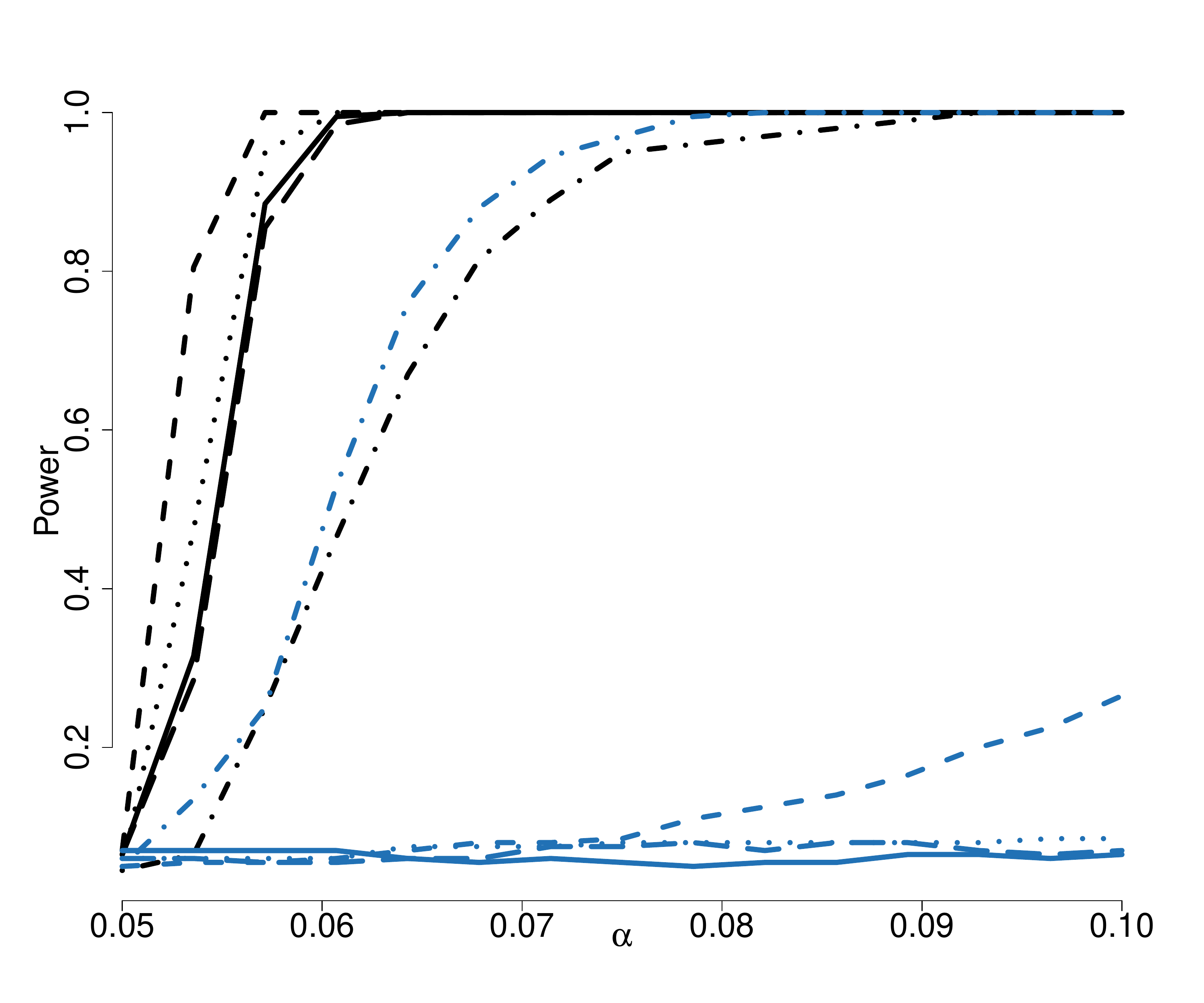}
    \caption*{(c)}
\end{minipage}
\begin{minipage}[c]{.49\textwidth} 
    \centering
    \includegraphics[width=0.9\textwidth]{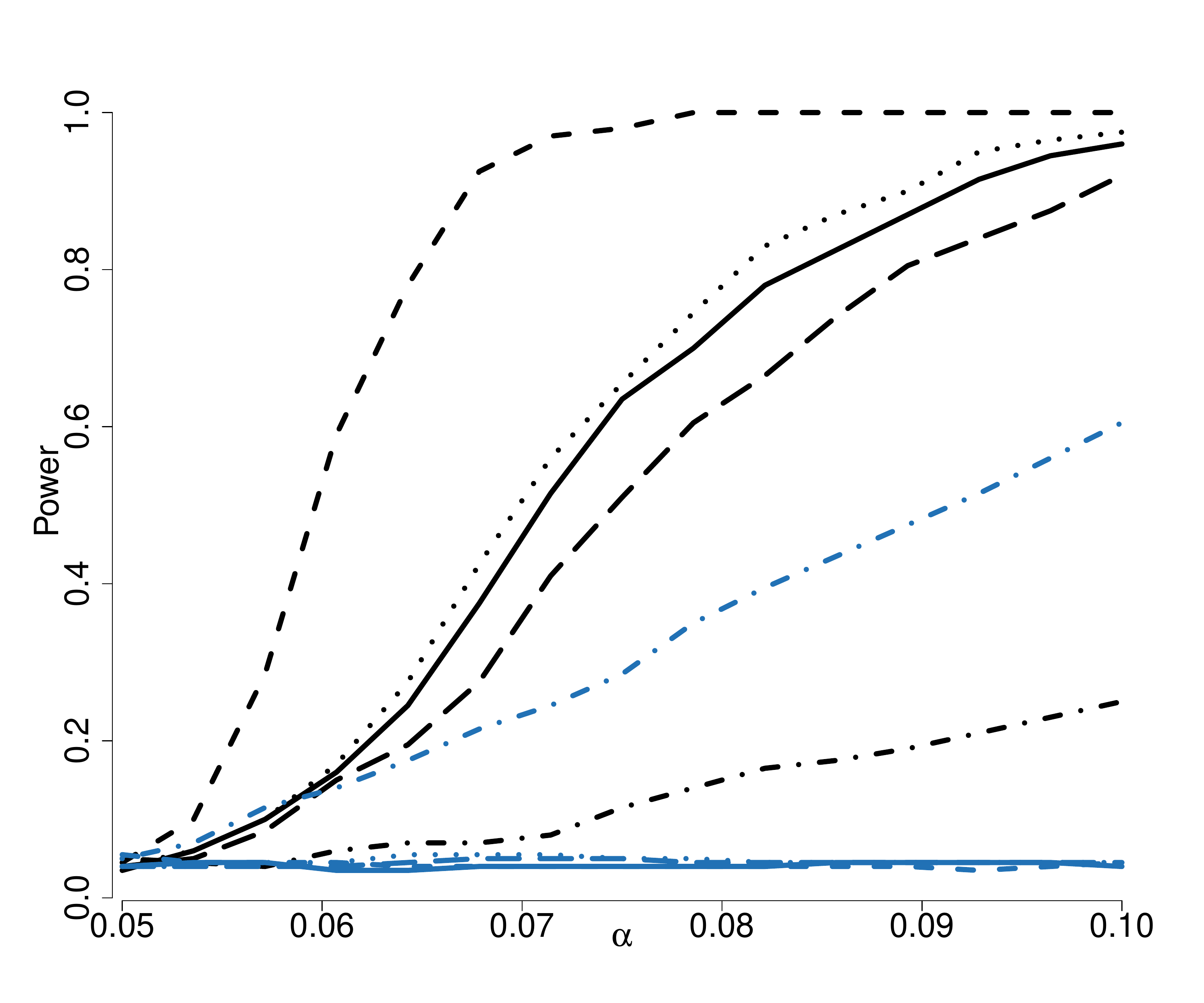}
    \caption*{(d)}
\end{minipage}
\caption{Empirical power curves for the tests based on $\widehat{\mathcal{W}}_{\scaledN}$, where the depth values were calculated with and without the derivatives. The underlying models were (a) $\mathcal{GP}$ samples, as $\beta$ varies with $\alpha=0.05$ (b) $t_1$ samples, as $\beta$ varies with $\alpha=0.05$ (c) $\mathcal{GP}$ samples, as $\alpha$ varies with $\beta=0.5$ (d) $t_1$ samples, as $\alpha$ varies with $\beta=0.5$. }
\label{fig::DvND}
\end{figure}
Figure \ref{fig::DvND} shows the empirical power curves for $N_j=250$ as $\beta$ and $\alpha$ vary under the Student-t and Gaussian processes. 
The skewed Gaussian results were similar to that of the Gaussian and so they are not included here. 
We can draw the following conclusions regarding the comparison of the FKWC tests with and without the derivatives. Including derivatives generally led to better performance, most notably when $\alpha$ differed between the two groups. Shape differences cannot generally be detected by the FKWC tests that do not include the derivatives; notice the flat curves in the bottom row of Figure \ref{fig::DvND}. The FKWC test with $\RP'$-based ranks appears to have the best performance. Similar relationships were seen across all sample sizes and under the percentile modification.

\begin{figure}
\begin{minipage}[c]{0.49\textwidth}
\centering
    \includegraphics[width=\textwidth]{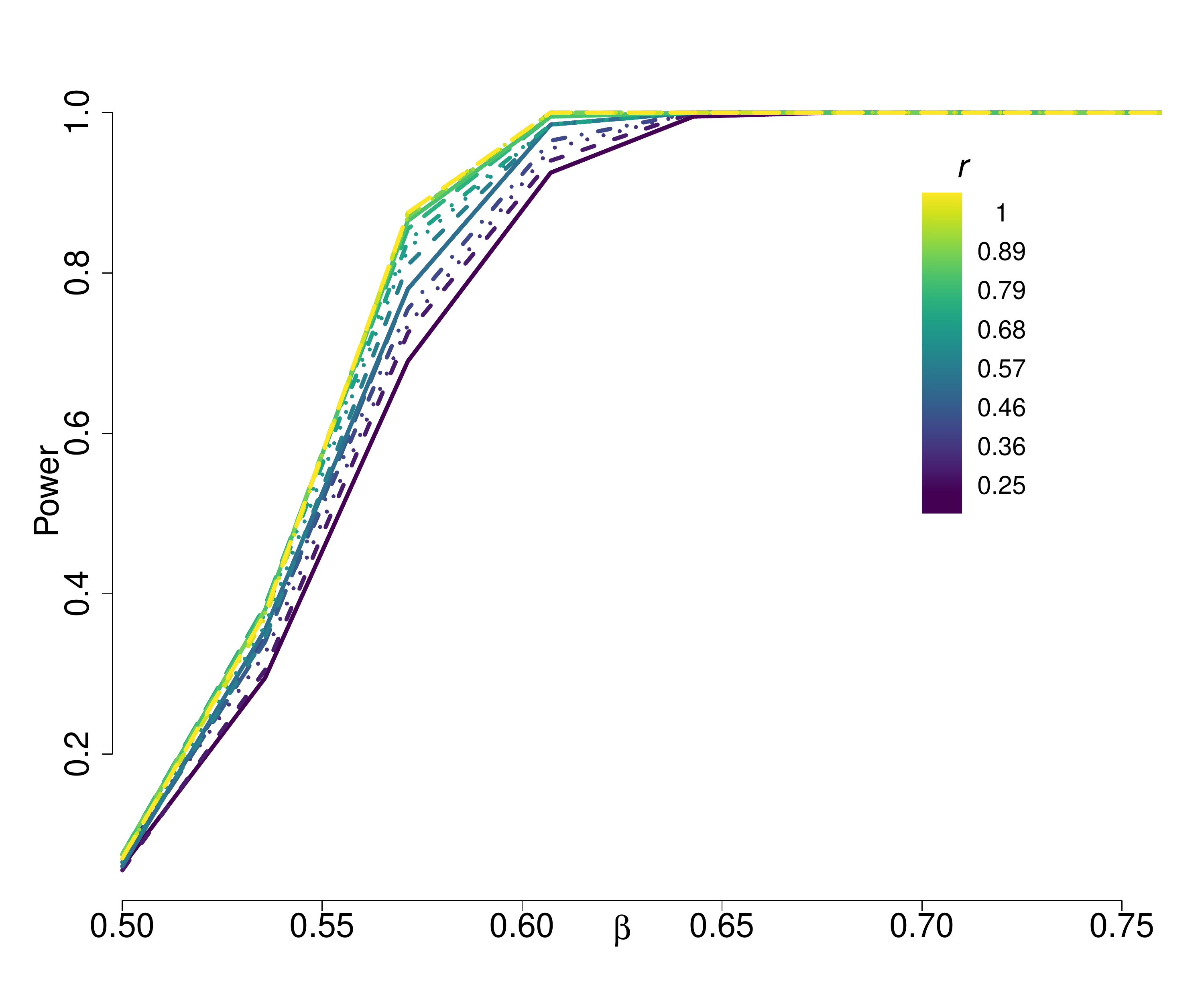}
       \caption*{(a)}
\end{minipage}
\begin{minipage}[c]{0.49\textwidth}
\centering
    \includegraphics[width=\textwidth]{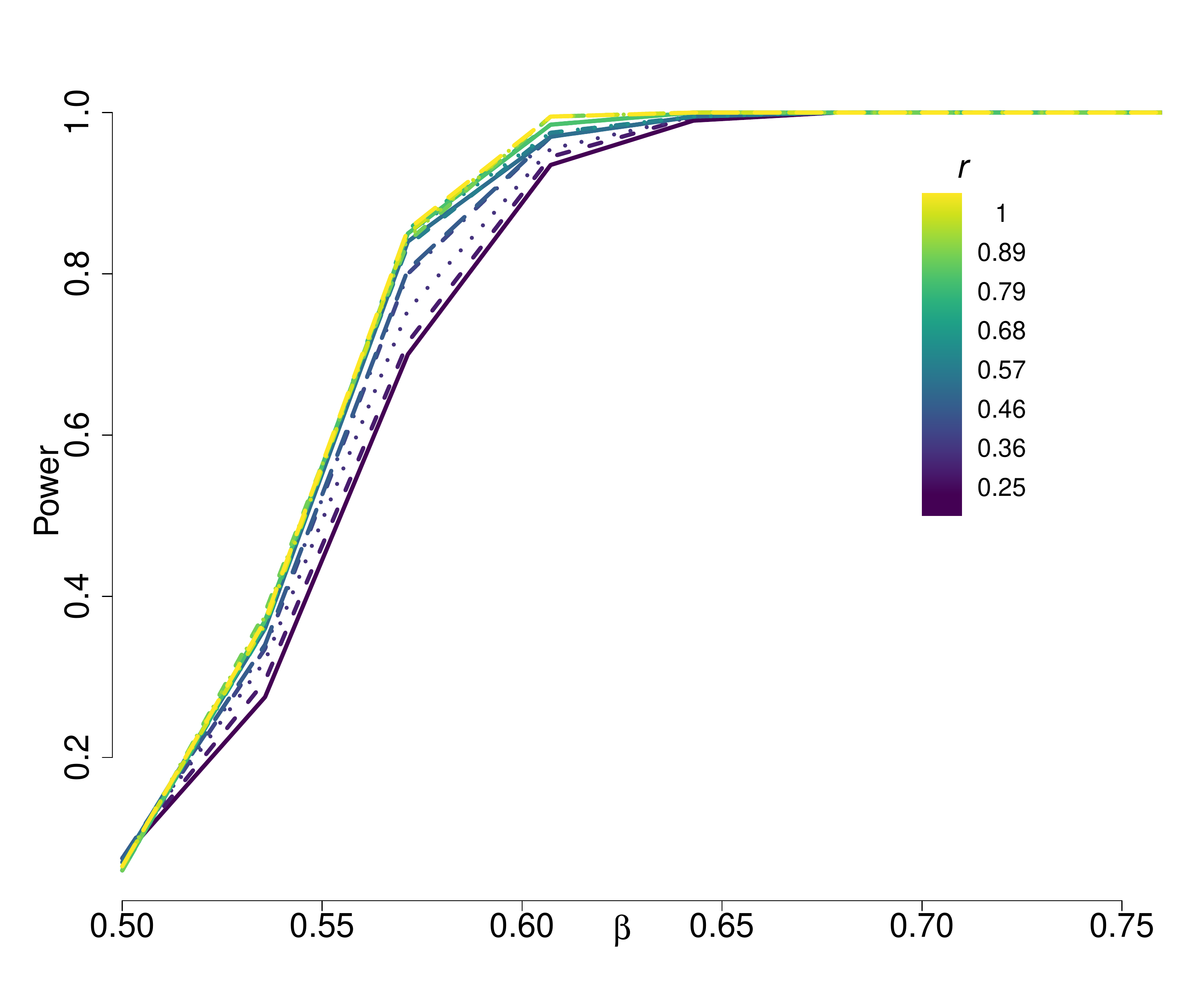}
    \caption*{(b)}
\end{minipage}
    \caption{Power curves for $\widehat{\mathcal{M}}_{\scaledN,r}$ for different values of $r$, as $\beta$ varies. $\RP'$ ranks were used to compute $\widehat{\mathcal{M}}_{\scaledN,r}$ and the underlying distributions were Gaussian (a) and Skewed Gaussian (b). Recall that $1-r$ is the percentage of central observations removed in the percentile modification. Note that there is a slight improvement over $\widehat{\mathcal{W}}_{\scaledN}$ in the Gaussian case but this effect vanishes under skewness.}
    \label{fig::PM_plot}  
\end{figure}

In terms of comparing $\widehat{\mathcal{M}}_{\scaledN,r}$ versus $\widehat{\mathcal{W}}_{N}$, in general there was not a large difference between the two tests. 
Figure \ref{fig::PM_plot} shows the power curves of $\widehat{\mathcal{M}}_{\scaledN,r}$ for different levels of $r$ as $\beta$ varies. 
$\RP'$ ranks were used to compute $\widehat{\mathcal{M}}_{\scaledN,r}$ and the underlying distributions were Gaussian (a) and Skewed Gaussian (b). 
It is useful to note that $\widehat{\mathcal{M}}_{\scaledN,1}=\widehat{\mathcal{W}}_{\scaledN}$. 
We observe that for $r\in (0.7,0.8)$ there is a slight improvement in power for the Gaussian case when the covariance differences between the samples are small. 
Outside of the Gaussian case, we observed that $\widehat{\mathcal{W}}_{\scaledN}$ was superior to $\widehat{\mathcal{M}}_{\scaledN,r}$. 
Similar results were seen at different sample sizes and different depth functions. 
We do not recommend using the percentile modification when the data is skewed or contains outliers, but it does provide some small additional power in the Gaussian case. 

We now compare the FKWC tests to some competing tests. 
As the results based on $\widehat{\mathcal{M}}_{\scaledN,r}$ are very similar to those of $\widehat{\mathcal{W}}_{\scaledN}$, we omit the results for $\widehat{\mathcal{M}}_{\scaledN,r}$. 
Similarly, graphically, we only compare the FKWC tests that include the derivative information. 
Table \ref{tab::null} contains the empirical sizes of each test and Figure \ref{fig:ARes} shows the power of the FKWC tests alongside the power of the competing tests. 
We can draw the following conclusions:
\begin{itemize}
    \item The sizes of the competing tests are slightly smaller than those of the FKWC tests when the processes are Gaussian and the the sample size is small ($N1=N2=50$).  
   \item The sizes of the FKWC tests are comparable to those of the competing tests when the processes are Skewed Gaussian. 
    \item The competing tests do not work well in the heavy tailed scenario. 
    \item Under these models, the FKWC tests have a higher power than the competing tests.
    \item The Boen test performs poorly in these conditions, perhaps because it is based on a finite number of principle components. 
\end{itemize}

\begin{figure}[t!]
\begin{minipage}[c]{.49\textwidth} 
    \centering
    \includegraphics[width=0.9\textwidth]{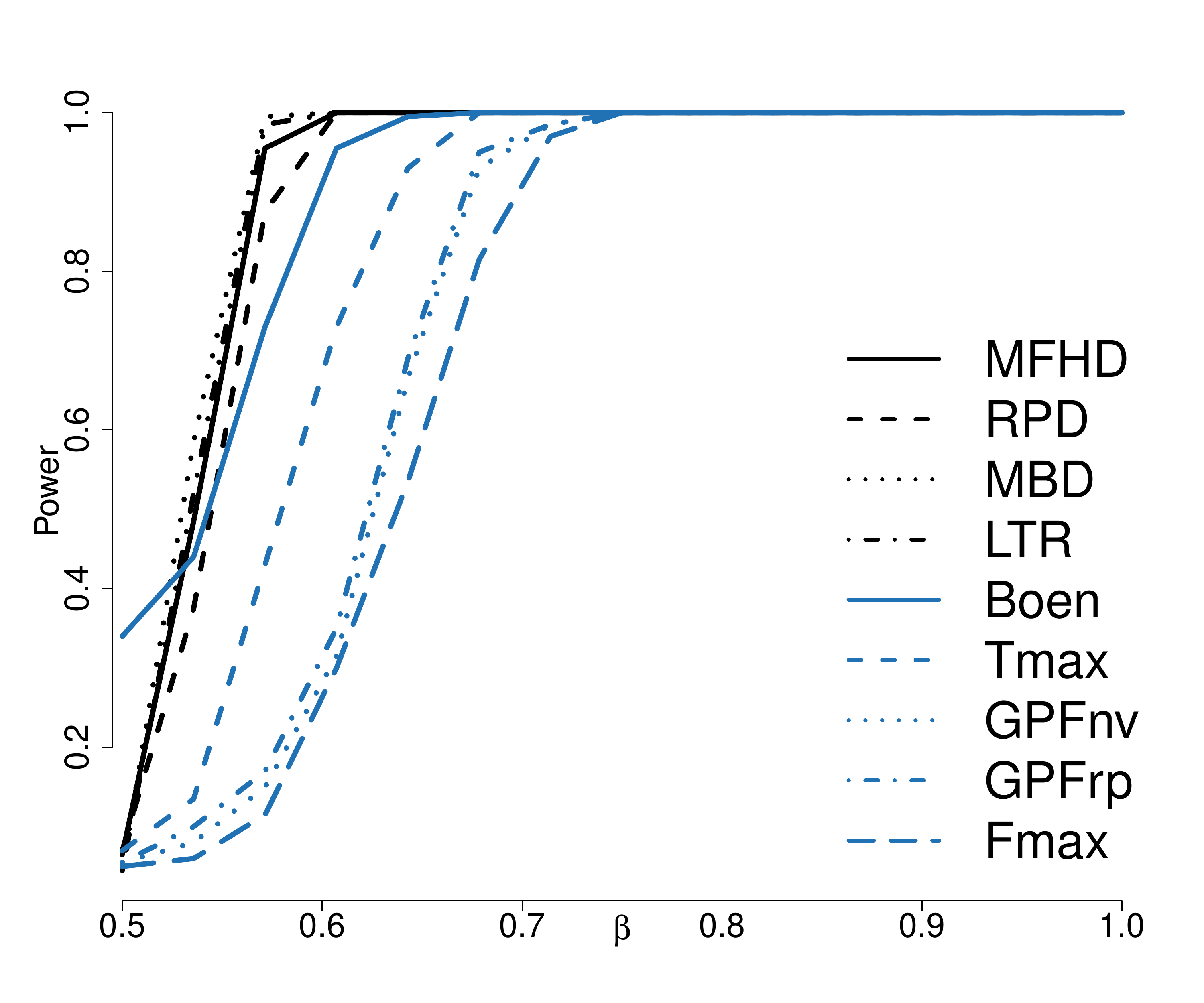}
     \caption*{(a) $\mathcal{GP}$; $\beta$ varies. }
\end{minipage}
\begin{minipage}[c]{.49\textwidth} 
    \centering
    \includegraphics[width=.9\textwidth]{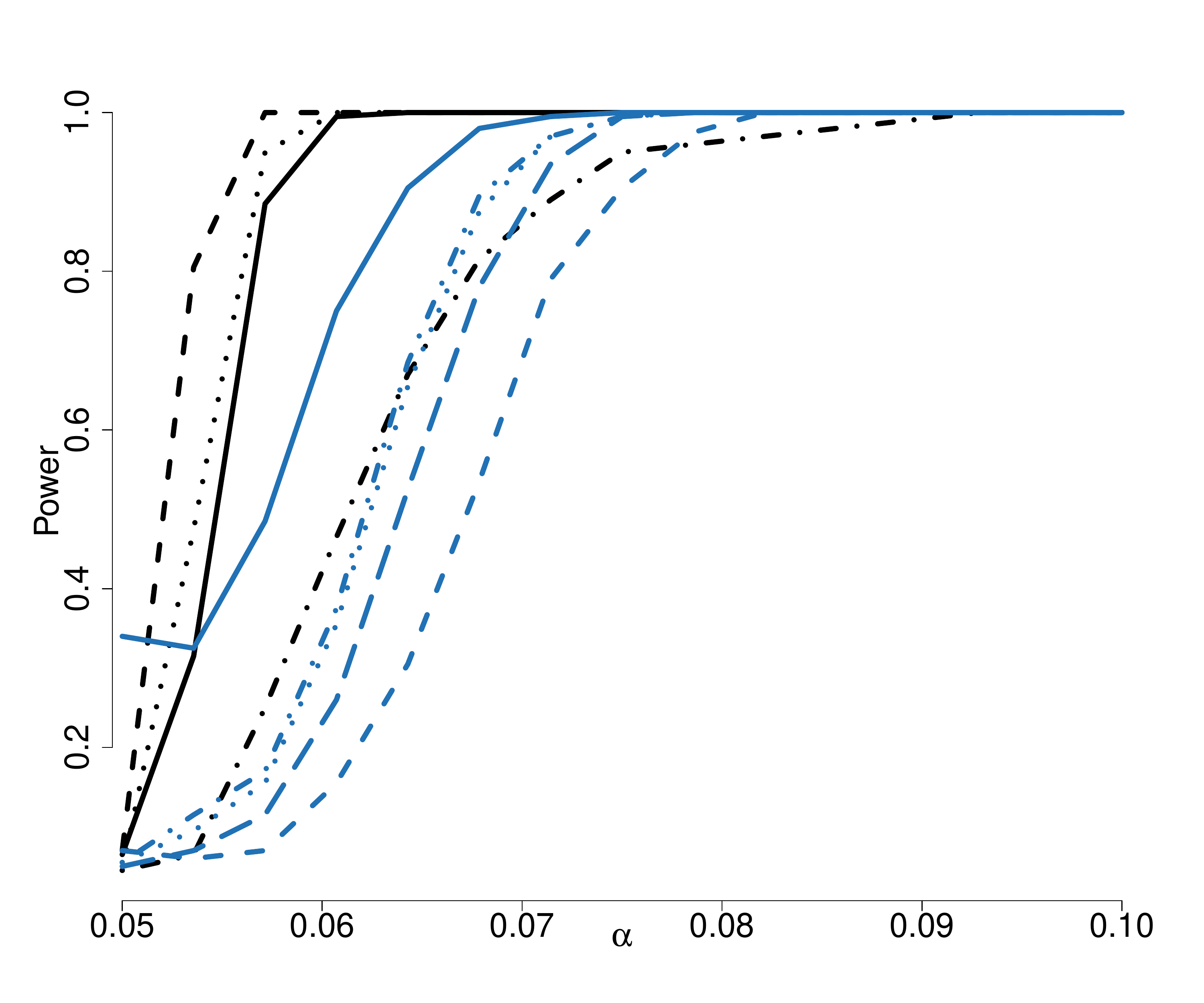}
    \caption*{(b) $\mathcal{GP}$; $\alpha$ varies. }
\end{minipage}
\hfill\newline
\begin{minipage}[c]{.49\textwidth} 
    \centering
    \includegraphics[width=.9\textwidth]{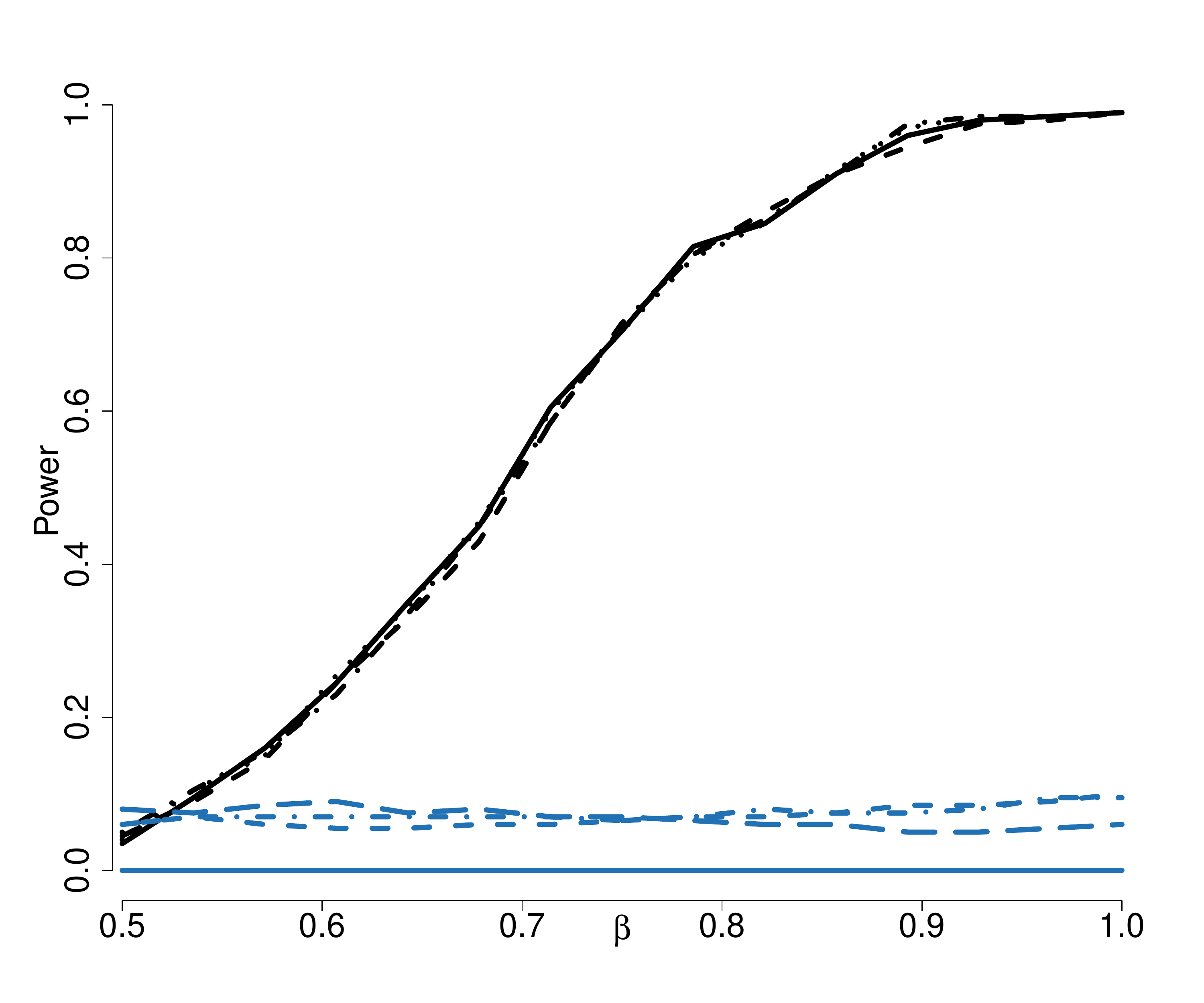}
     \caption*{(c) $t_1$; $\beta$ varies.}
\end{minipage}
\begin{minipage}[c]{.49\textwidth} 
    \centering
    \includegraphics[width=0.9\textwidth]{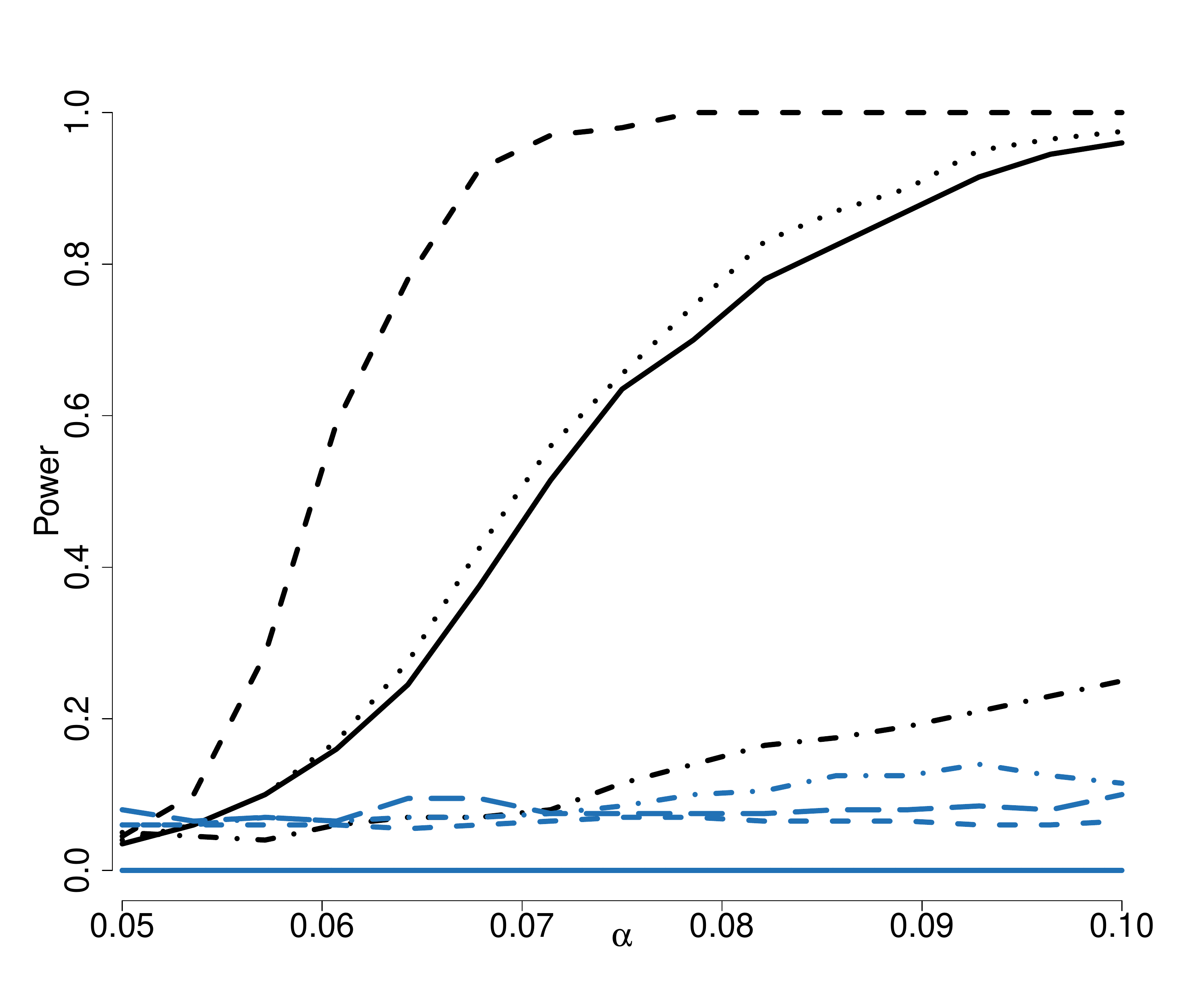}
     \caption*{(d) $t_1$; $\alpha$ varies.}
   \end{minipage}
   \hfill\newline
\begin{minipage}[c]{.49\textwidth} 
    \centering
    \includegraphics[width=.9\textwidth]{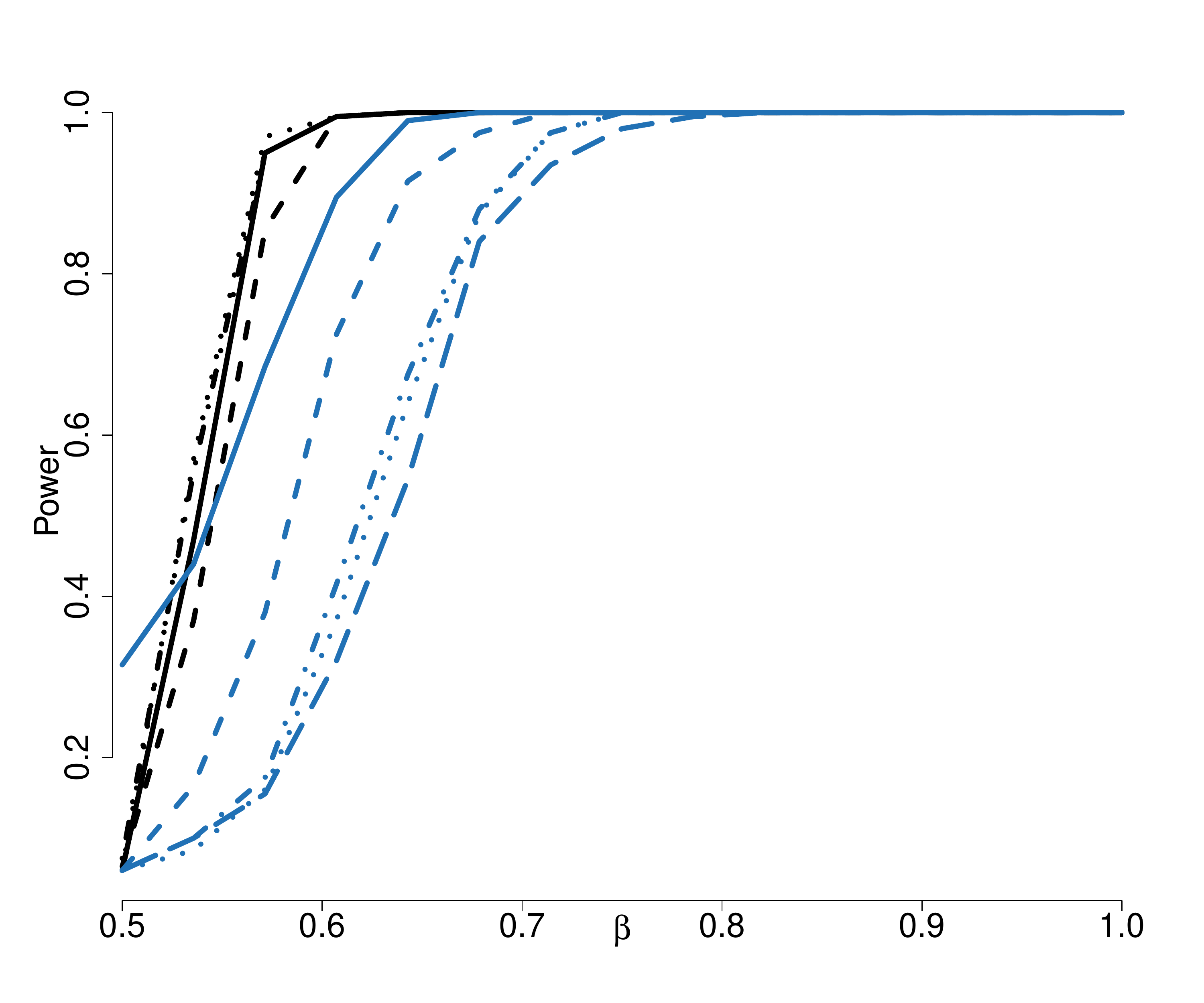}
    \caption*{(e) $\mathcal{SG}$; $\beta$ varies.}
\end{minipage}
\begin{minipage}[c]{.49\textwidth} 
    \centering
    \includegraphics[width=0.9\textwidth]{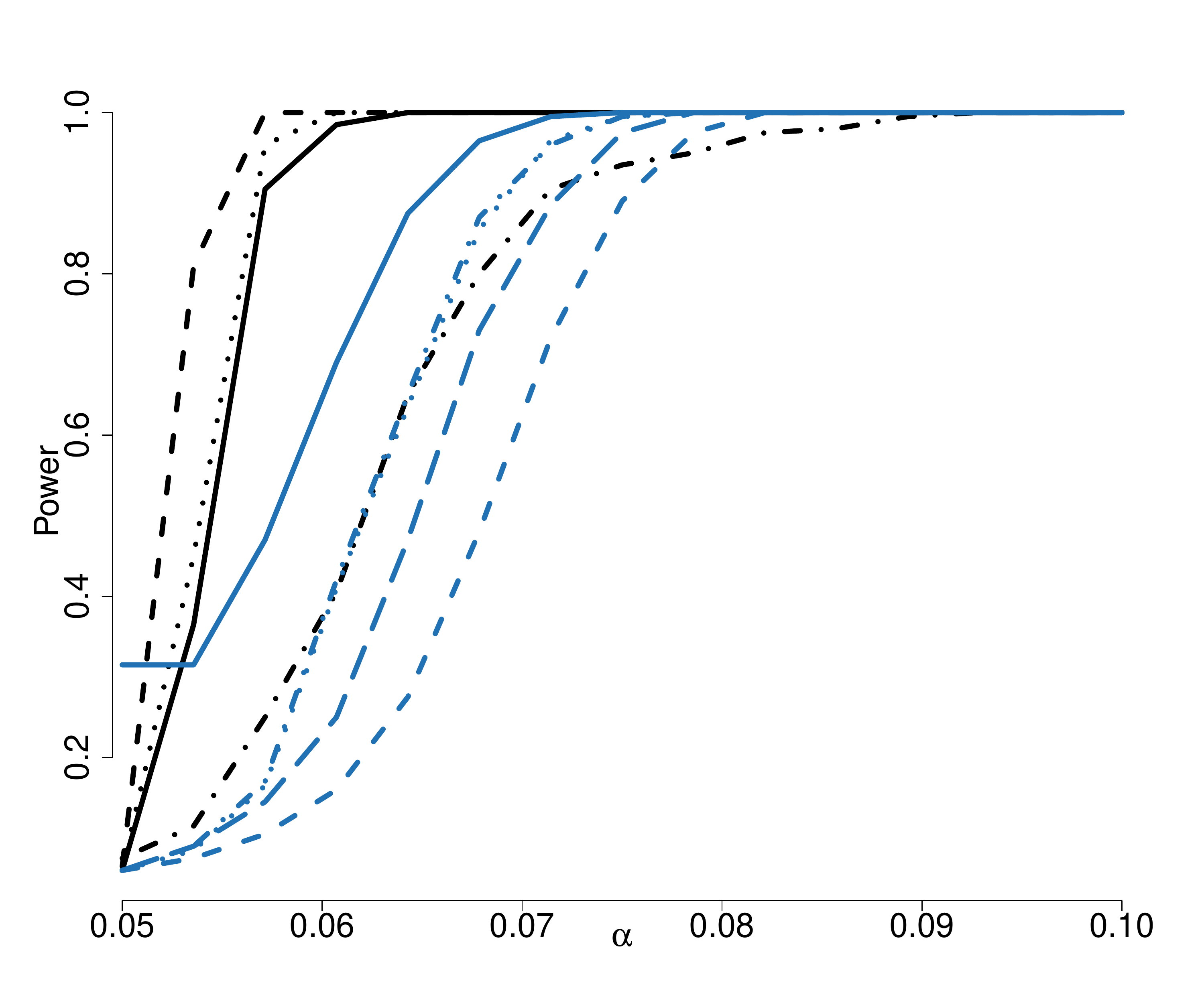}
     \caption*{(f) $\mathcal{SG}$; $\alpha$ varies.} 
   \end{minipage}
  \caption{Power curves as the $\beta$ parameter (left column) and $\alpha$ parameter (right column) of the second sample moves away from the null hypothesis. Here, $N_j=250$. }
  \label{fig:ARes}
\end{figure}
\begin{table}[t]
\centering
\begin{tabular}{cl|rcc|ccc|ccc}
  \hline
  &\multirow{2}*{ Test}& \multicolumn{3}{c}{Gaussian} & \multicolumn{3}{c}{Student t}  & \multicolumn{3}{c}{Skewed Gaussian} \\ 
 \cline{3-11}
  & & 50 & 100 & 250 & 50 & 100 & 250 & 50 & 100 & 250\\
  \hline
\multirow{ 8}{*}{FKWC} &$\MFHD'$ & 0.06 &\textbf{ 0.03} & 0.06 & 0.08 & 0.06 &\textbf{0.04}& 0.08 & 0.07 &  0.06 \\
&$\RP_{20}'$ & 0.06 &\textbf{ 0.03} & 0.07 & 0.06&\textbf{0.04}&\textbf{0.04}& 0.07 &\textbf{0.04}&  0.06 \\ 
&$\MBD'$ & \textbf{0.05} & \textbf{0.02} &\textbf{0.04}& 0.08 & 0.06 &\textbf{0.04}& 0.06 & 0.06 &  0.06  \\ 
&$\LTR'$ & 0.06 & \textbf{0.04} &\textbf{0.04}& 0.08 & 0.06 & \textbf{0.05} & 0.07 & 0.06 & 0.07 \\ 
&$\MFHD$ & 0.06 & \textbf{0.02} & 0.06 & 0.08 & \textbf{0.05} &\textbf{0.04}& 0.08 & 0.06 & 0.08 \\ 
&$\RP_{20}$ &\textbf{0.04}& \textbf{0.04} &\textbf{0.04}& 0.07 & 0.06 &\textbf{0.04}& 0.06 & \textbf{0.05} & 0.07 \\ 
&$\MBD$ & \textbf{0.05} & \textbf{0.02} & 0.07 & 0.08 & 0.06 &\textbf{0.04}& 0.07 & 0.07 & 0.09\\
&$\LTR$ & 0.07 & 0.06& \textbf{0.05} & 0.08 & 0.06 & 0.06 & 0.06 & 0.06 & 0.08 \\ 
\hline
\multirow{ 8}{*}{Competing} &  Boen & 0.20 & 0.32 & 0.40 & \textbf{0.00} & \textbf{0.00} & \textbf{0.00}& 0.22 & 0.18 & 0.28 \\
 & L2nv & \textbf{0.02} & \textbf{0.02} &0.06 & 1.00 & 1.00 & 1.00 & \textbf{0.02} & 0.08 &\textbf{ 0.00}\\ 
 & L2br  & \textbf{0.02} & \textbf{0.02} & 0.06 & 1.00 & 1.00 & 1.00 & \textbf{0.04} & 0.10 &\textbf{ 0.00}\\ 
 & L2rp & \textbf{0.02} & \textbf{0.04} &\textbf{0.04}& 0.06 & 0.10 &\textbf{0.04}& \textbf{0.04} & 0.10 &\textbf{ 0.00}\\ 
 & Tmax &\textbf{0.04}& \textbf{0.02} &\textbf{0.04}& 0.10 &\textbf{0.04} &\textbf{0.04}& 0.12 & \textbf{0.02} & \textbf{0.04} \\ 
&  GPFnv  & \textbf{0.02} & \textbf{0.02} &\textbf{0.04}&\textbf{ 0.00}&\textbf{ 0.00}&\textbf{ 0.00}& \textbf{0.02} & 0.08 &\textbf{ 0.00}\\
 & GPFrp &\textbf{0.04}& \textbf{0.02} &\textbf{0.04}& 0.06 & 0.08 & 0.06 & 0.06 & 0.08 &\textbf{ 0.00}\\ 
&  Fmax & 0.08 & 0.06 & \textbf{0.02} & 0.06 & \textbf{0.02} & \textbf{0.02} & \textbf{0.02} & 0.06 &\textbf{ 0.00}\\ 
   \hline
\end{tabular}
\caption{Empirical sizes for $J=2$ for different tests under the infinite dimensional models. The first row indicates the underlying process and the second row indicates the sample size of each group. }
\label{tab::null}
\end{table}

We now analyse the simulation results from the finite dimensional models. 
In this simulation model, the data were simulated from a Gaussian process where we directly specified $K$ non-zero eigenvalues of the covariance operator. 
A Fourier basis was used for the eigenfunctions. 
This setup was used to do two things: assess the performance of the test when the underlying model had a low dimension and compare additional kinds of covariance differences, such as differences where the trace norms of the covariance operators of each group are equal. 
Here, only the two sample case was checked, and we used $N_1=N_2=50$ and 100. 
Specifically, we ran six scenarios, which, letting $\lambda_{jk}$ be the eigenvalues of the covariance operator of group $j$, are
\begin{enumerate}
    \item Reversed short linear decay:  $\lambda_{1k}=k,\ \lambda_{2k}=3-k+1,\ k<4,\ \lambda_{jk}=0,\ k\geq 4$.
    \item Reversed long linear decay: $\lambda_{1k}=k,\ \lambda_{2k}=11-k+1,\ k<12,\ \lambda_{jk}=0,\ i\geq 12$.
    \item Reversed long exponential decay: $\lambda_{1k}=2^k,\ \lambda_{2k}=2^{11-k+1},\ k<12,\ \lambda_{jk}=0,\ k\geq 12$.
    \item Scaled short linear decay: $\lambda_{1k}=k,\  \lambda_{2k}=1.5\lambda_{1k},\ k<4,\ \lambda_{jk}=0,\ k\geq 4 $.
    \item Scaled long linear decay: $\lambda_{1k}=k,\ \lambda_{2k}=1.5\lambda_{1k}\,\ k<12,\ \lambda_{jk}=0,\ k\geq 12$.
    \item Scaled long exponential decay: $\lambda_{1k}=2^k,\ \lambda_{2k}=1.5\lambda_{1k},\ k<12,\ \lambda_{jk}=0,\ k\geq 12$.
\end{enumerate}
Table \ref{tab::evres} contains the empirical power of each test, for each of the finite dimensional scenarios. 
Notice first that the competing tests perform better in these low dimensional settings. 
This could be due to the fact that the lower complexity allowed for easier approximation of the null distributions. 
Secondly, we note that in scenarios where there was no difference in the trace of the covariance, the $\LTR$ depth could not detect the difference. 
This is in agreement with Theorem \ref{thm::SD}. 
Incorporating the derivatives again shows a large improvement in performance of the FKWC tests. 
Also, as above the random projection depth with the derivatives performs almost as well as the competing tests. 
In conclusion, the simulation has shown that the FKWC tests are very competitive with other existing tests. 
This is especially true when the underlying data have many positive eigenvalues and/or the data has heavy tails, skewness or outliers. 
Specifically, the test based on the $\RP$ ranks performed the best, and therefore we recommend using this depth function to execute the test in practice. 
\begin{table}[t]
\centering
\begin{tabular}{llcccccc}
& \multicolumn{6}{c}{Simulation Scenario}\\
  \hline
 Test  & 1 & 2 & 3 & 4 & 5 & 6 \\ 
  \hline
$\MFHD'$  & 0.38 & \textbf{1.00} & \textbf{1.00} & 0.99 & \textbf{1.00} & \textbf{1.00} \\ 
  $\RP_{20}'$ & 0.99 & \textbf{1.00} & \textbf{1.00} & 0.93 & \textbf{1.00} & 0.99 \\ 
 $\MBD'$ & 0.80 & \textbf{1.00} & \textbf{1.00} & \textbf{1.00} & \textbf{1.00} & \textbf{1.00} \\ 
$\LTR'$ & 0.07 & \textbf{1.00} & \textbf{1.00} & \textbf{1.00} & \textbf{1.00} & 0.99 \\ 
$\MFHD$  & 0.61 & 0.07 & 0.44 & 0.87 &\textbf{1.00}& 0.90 \\ 
 $\RP_{20}$& 0.24 &\textbf{1.00}&\textbf{1.00}& 0.96 &\textbf{1.00}& 0.97 \\ 
   $\MBD$ & 0.06 & 0.07 & 0.05 & 0.87&\textbf{1.00}&0.9\\ 
   $\LTR$ & 0.36 & 0.10 & 0.06 & 0.92 &\textbf{1.00} & 0.93 \\ 
  Competing& \textbf{1.00} &\textbf{1.00}&\textbf{1.00}&\textbf{1.00}&\textbf{1.00}&\textbf{1.00}\\ 
   \hline
\end{tabular}
\caption{Eigenvalue simulation scenarios when $N_1=N_2=100$. The first row indicates the scenario number. ``Competing' stands for the competing tests; all competing tests had the same empirical power. The symbol ' indicates the depth calculation included the derivatives of the observed curves as discussed in Section \ref{sec::dep}.  }
\label{tab::evres}
\end{table}

\section{Applications to real data}\label{sec::dataana}
In this section we present an application of our methodology to two different functional datasets. 
One is comprised of intraday stock prices and the other is comprised of digitized speech. 
\subsection{F-GARCH residual analysis of intraday stock price curves}\label{sec::dataana1}
We analyse the daily asset price curves of $J=3$ different stocks (\texttt{twtr}, \texttt{fb} and \texttt{snap}) starting on June 24th 2019 and ending March 20th 2020, which gives $N_1=207$ and $N_2=N_3=208$. 
Precisely, for each stock the price was measured over the course of the trading day in one minute intervals, for a total of 390 minutes per day. 
In order to account for edge effects from smoothing the curves, we trimmed 10\% of the minutes from the beginning of the day and 5\% of the minutes from the end of the day.  
This resulted in 332 minutes of stock prices. 
In actuality, we analysed the log returns, viz.
$$X_{ji}(t)=\ln(Y_{ji\floor{331t}+1})-\ln(Y_{ji\floor{331t}}),$$
where $Y_{jik}$ is the $j^{th}$ asset price on the $i^{th}$ day at minute $k$. 
Figure \ref{fig:DA}\textcolor{blue}{(a)} shows the intraday log return curves $\{X_{2i}(t)\}_{i=1}^{208}$ of Facebook (\texttt{fb}) stock. 
The data was fit to a B-spline basis, using 50 basis functions, see \texttt{smooth.basis} in the \texttt{fda} \texttt{R} package. 

Notice that the magnitudes of the curves vary widely. 
Figure \ref{fig:DA}\textcolor{blue}{(b)} displays the squared norms of the daily curves as a function of the day $i$. 
We can see that the magnitude of each observation is related to the day on which it was observed. 
For example, around March 2020 the norms are higher, likely due to the volatility which resulted from the COVID-19 pandemic.

\begin{figure}[t!]
\begin{minipage}[c]{.5\textwidth} 
    \centering
    \includegraphics[width=0.95\textwidth]{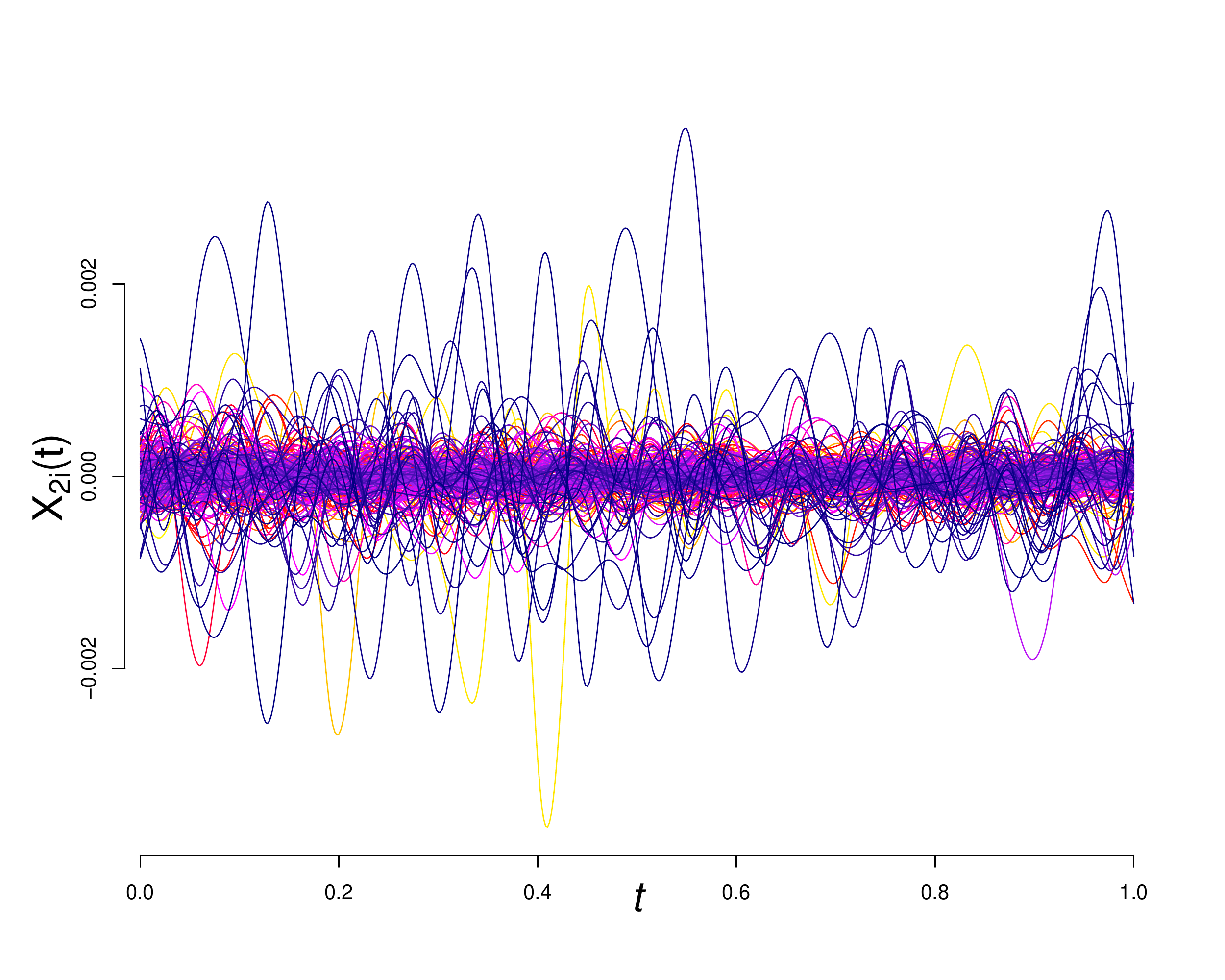}
    \caption*{(a)}
\end{minipage}
\begin{minipage}[c]{.5\textwidth} 
    \centering
    \includegraphics[width=.95\textwidth]{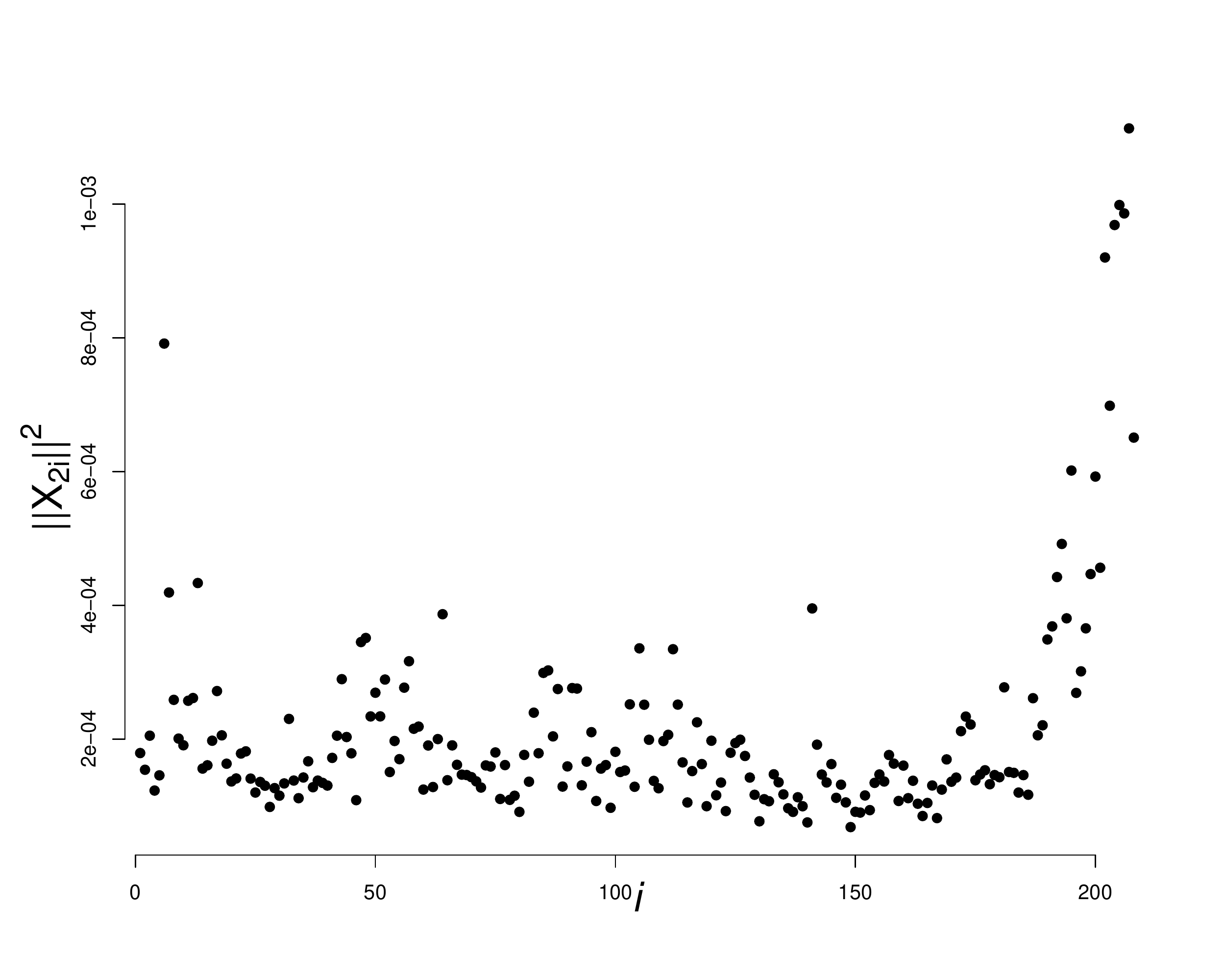}
    \caption*{(b)}
\end{minipage}
\hfill\newline
\begin{minipage}[c]{.5\textwidth} 
    \centering
    \includegraphics[width=.95\textwidth]{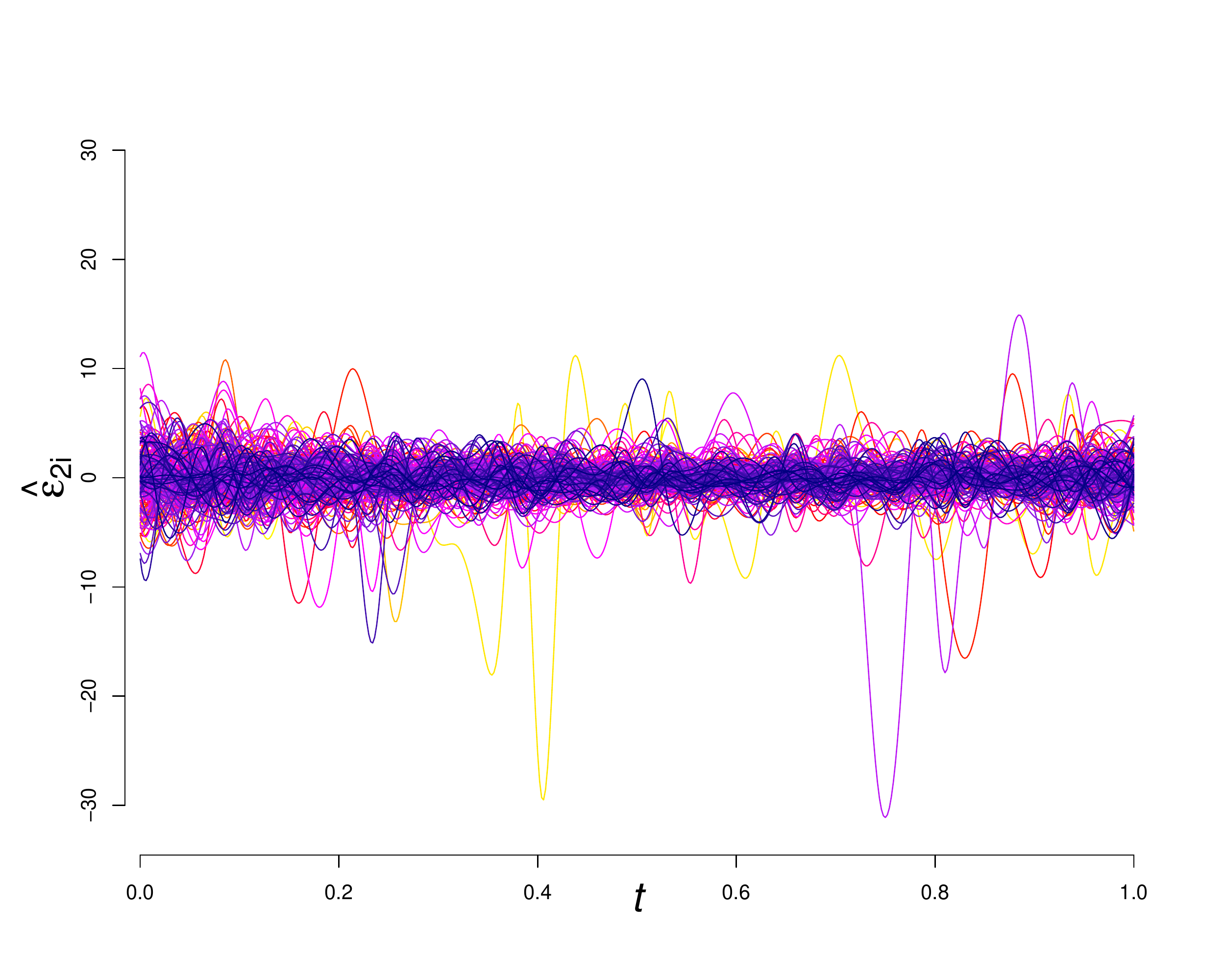}
    \caption*{(c)}
\end{minipage}
\begin{minipage}[c]{.5\textwidth} 
    \centering
    \includegraphics[width=0.95\textwidth]{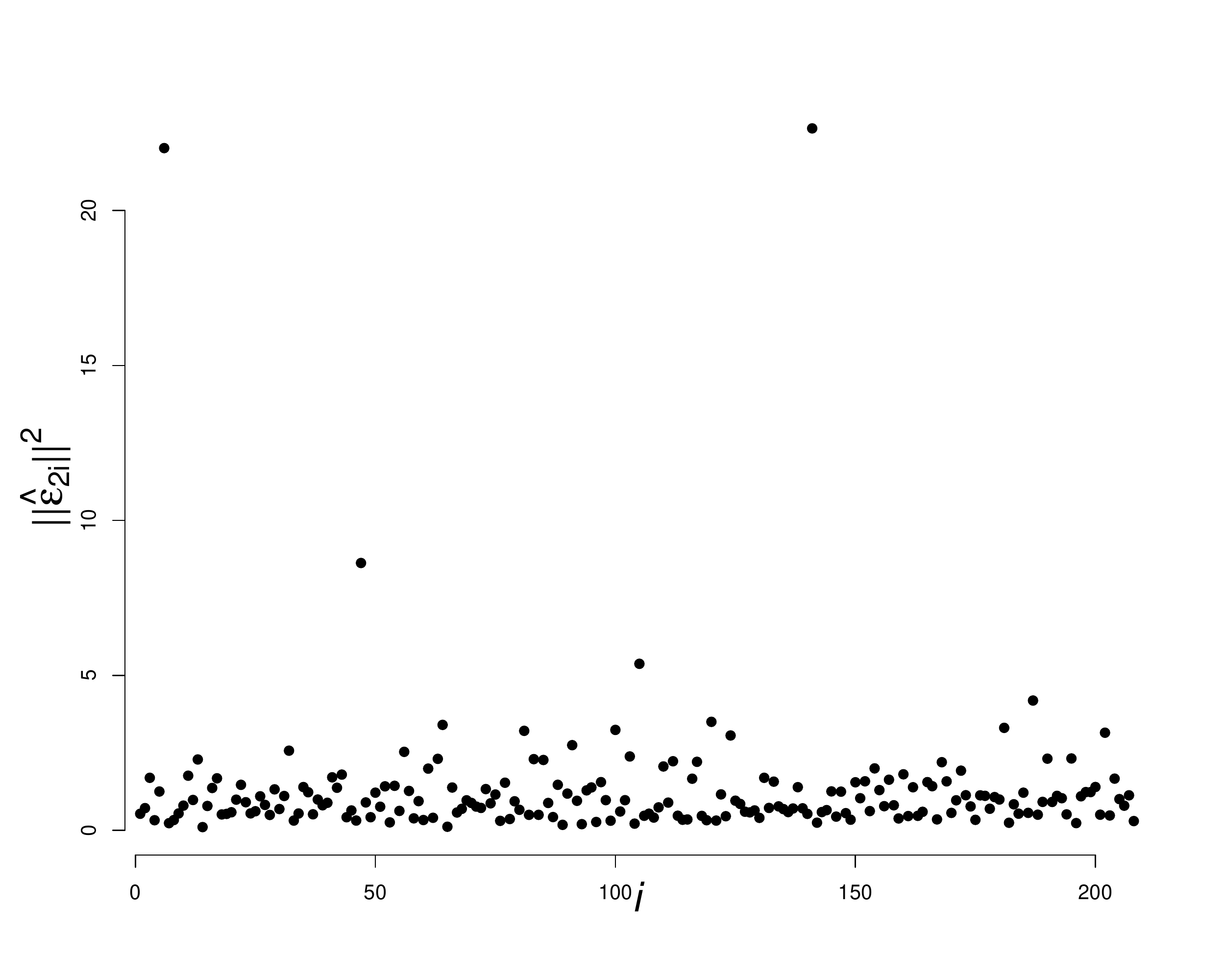}
      \caption*{(d)}
    \label{fig:Snap}
\end{minipage}
  \caption{(a) Daily log differenced intraday return curves for \texttt{fb} stock, starting on June 24th 2019 and ending March 20th 2019. (b) Daily squared norms of the intraday returns. Notice that these norms vary with the time period; the curves exhibit heteroskedastic features. For example, the most recent month of returns are much more variable. (c) Residuals for the \texttt{fb} log returns after fitting a functional GARCH(1,1) model (d) squared norms of the residuals over time. }
  \label{fig:DA}
\end{figure}
\begin{figure}[t!]
\begin{minipage}[c]{.31\textwidth} 
    \centering
    \includegraphics[width=\textwidth]{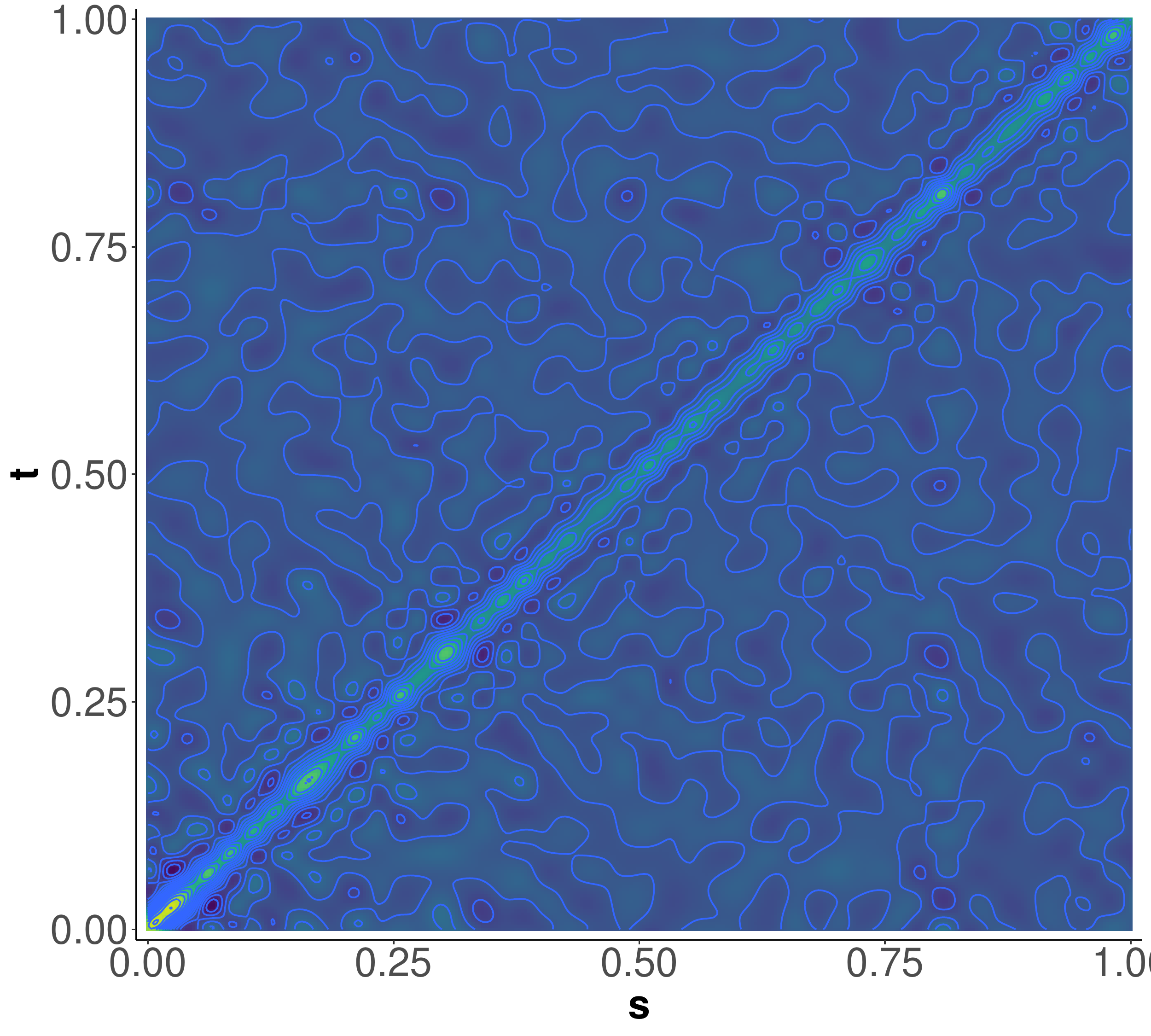}
    \caption*{(a)}
\end{minipage}
\begin{minipage}[c]{.31\textwidth} 
    \centering
    \includegraphics[width=\textwidth]{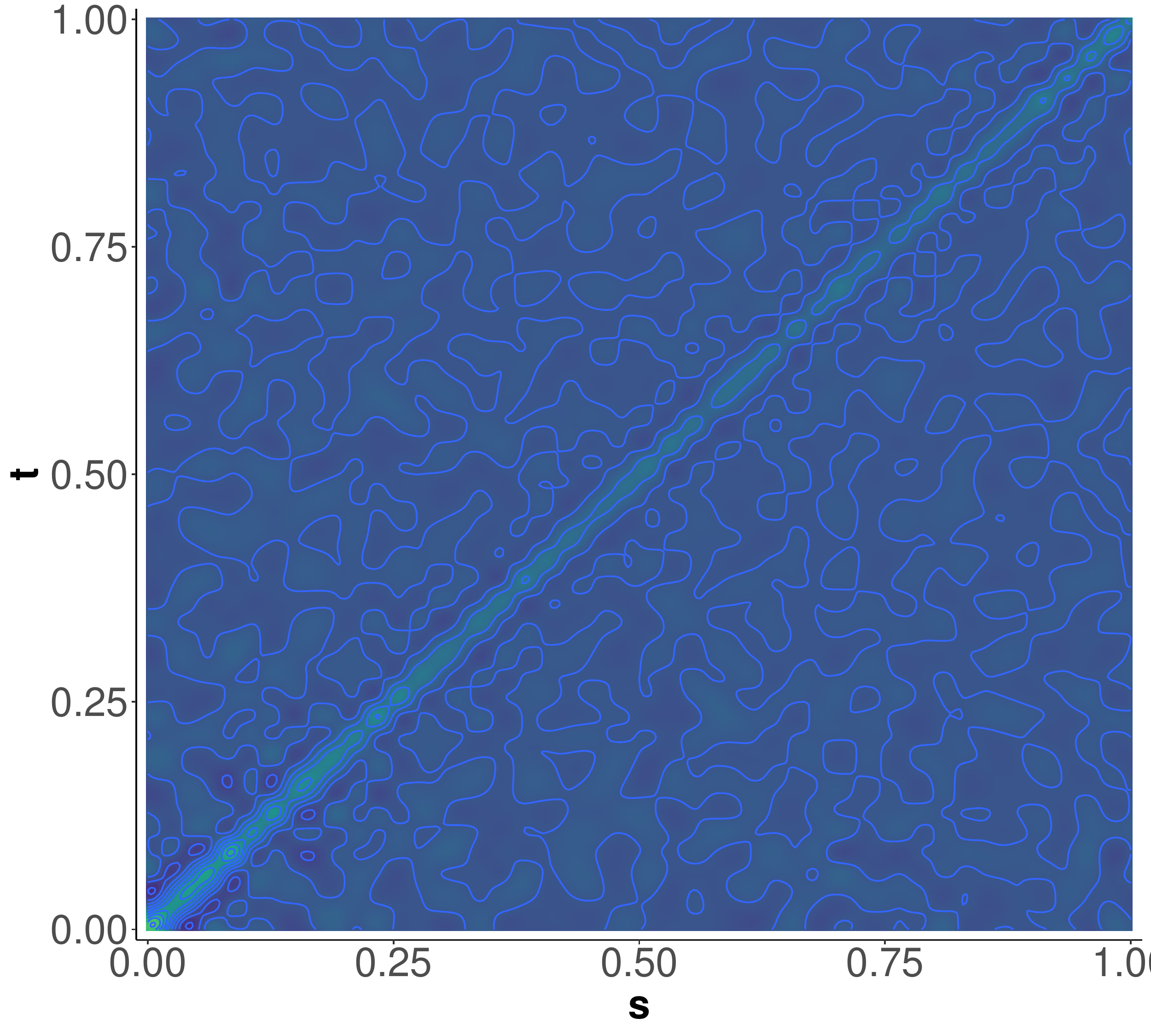}
    \caption*{(b)}
\end{minipage}
\begin{minipage}[c]{.35\textwidth} 
    \centering
    \includegraphics[width=\textwidth]{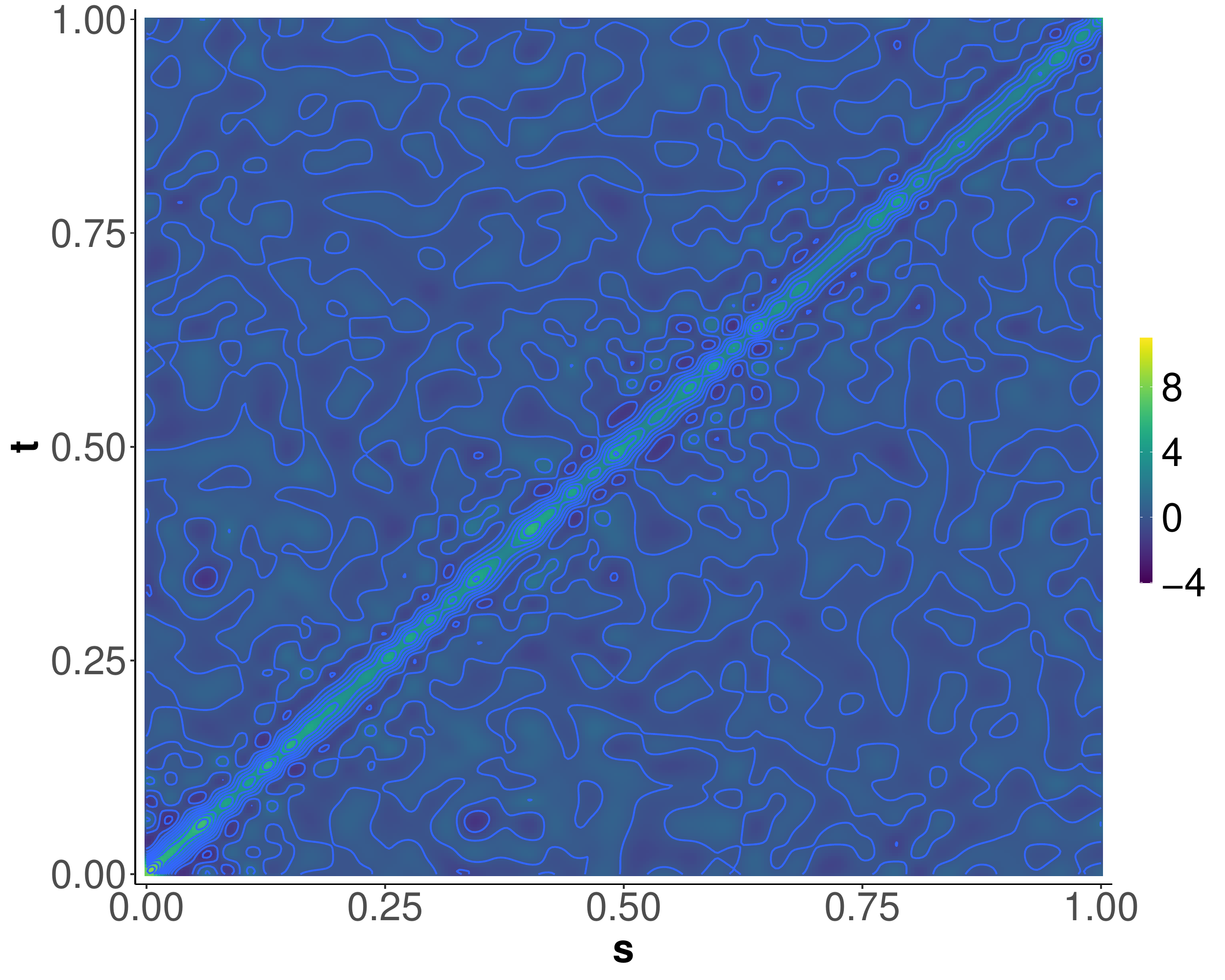}
    \caption*{(c)}
\end{minipage}
  \caption{Covariance kernels $\kerr(s,t)$ of the residuals of (a) \texttt{twtr},  (b) \texttt{fb} and (c) \texttt{snap} log returns. }
  \label{fig:DA2}
\end{figure}
To handle the heteroskedasticity and any serial correlation present in the data, we employ a functional GARCH(1,1) model \citep{Aue2017, CEROVECKI2019353} and apply the FKWC test to the residuals. 
The idea is to decompose the data into the conditional volatility $\eta^2_i$ and the independent error $\epsilon_{ji}(t)$, which can be approximated by the residuals $\widehat{\epsilon}_{ji}(t)$. 
Unlike in the univariate GARCH model, the second order behaviour of $\epsilon_{ji}(t)$ can differ between different assets; $\E{}{\epsilon^2_{ji}(t)}=1$ for all $t$ is assumed for an identifiable model \citep{CEROVECKI2019353} but nothing is assumed about $\E{}{\epsilon_{ji}(t)\epsilon_{ji}(s)}$ for $s\neq t$. 
Thus, it is also of interest to investigate the properties of $\E{}{\epsilon_{ji}(t)\epsilon_{ji}(s)}$. 
For example, if the errors come from the same distribution, then the residuals can be pooled and bootstrapped to provide standard errors. 

Since this type of data is typically heavy tailed, a robust test is suitable. 
In order to check the condition that $\eta^2_{ji}$ completely encapsulates the serial dependence in the data, we use the tests described by \cite{Rice2019b}. 
Specifically, we fit the functional GARCH(1,1) to each series of intraday returns using quasi-maximum likelihood \citep{CEROVECKI2019353}. 
We assumed that the volatility curves could be represented as linear combinations of $M$ Bernstein basis functions. 
These were chosen based on a combination of the Box-Jenkins type test for the functional GARCH model \citep{Rice2019b}, assessing the fit of the raw mean of the squares graphically (see \cite{CEROVECKI2019353}) and keeping the number of basis functions similar between assets. 
This resulted in choosing $M=4$, and our results were insensitive to the number of basis functions in terms of testing the residuals for a difference in covariance. 
Figure \ref{fig:DA}\textcolor{blue}{(c)} shows the resulting residuals of the GARCH(1,1) model fitted to the \texttt{fb} stock log returns and Figure \ref{fig:DA}\textcolor{blue}{(d)} shows the norms of those residuals as a function of the day $i$. 
Notice that both the residuals and their norms are fairly uniform over time, especially when compared with the raw data. 
Figure \ref{fig:DA}\textcolor{blue}{(d)} also shows that there are some outliers in the data.  

Figure \ref{fig:DA2} shows contour plots of the estimated covariance kernels of the residuals of each functional time series, where 5\% of the lowest random projection depth (with derivatives) observations were trimmed to account for the outliers. 
Notice that the estimated covariance kernels of the residuals of \texttt{fb} differs from the other two assets visually.  
We conducted the FKWC test at the 5\% level of significance, using ranks based on the random projection depth which incorporates the derivatives. 
The means of the ranks are 244.4203 , 424.2837, and  266.9712 for the \texttt{twtr}, \texttt{fb},  and \texttt{snap} residuals, respectively; $\widehat{\mathcal{W}}_{\scaledN}=120.37$ and we reject the hypothesis that these three series have the same covariance kernels. 
The means of the ranks are similar for that of the \texttt{twtr} and \texttt{snap} stock, but the \texttt{fb} stock differs, which matches Figure \ref{fig:DA2}. 
\subsection{Comparing speech variability with phoneme periodograms}\label{sec::dataana2}
In this section we analyse the Phoneme data, where the observations are log peroidograms of digitized speech. 
The data can be retrieved as part of the \texttt{fda R} package \citep{Hastie1995}. 
The data is split into five groups representing the syllables  `aa', `ao', `dcl', `iy' and  `sh'. 
The goal is to characterize differences between the syllables' distributions in order to aid understanding of speech as well as to help improve the performance of speech recognition models. 
\begin{figure}[t]
\begin{minipage}[c]{.5\textwidth}
        \centering
    \includegraphics[width=\textwidth]{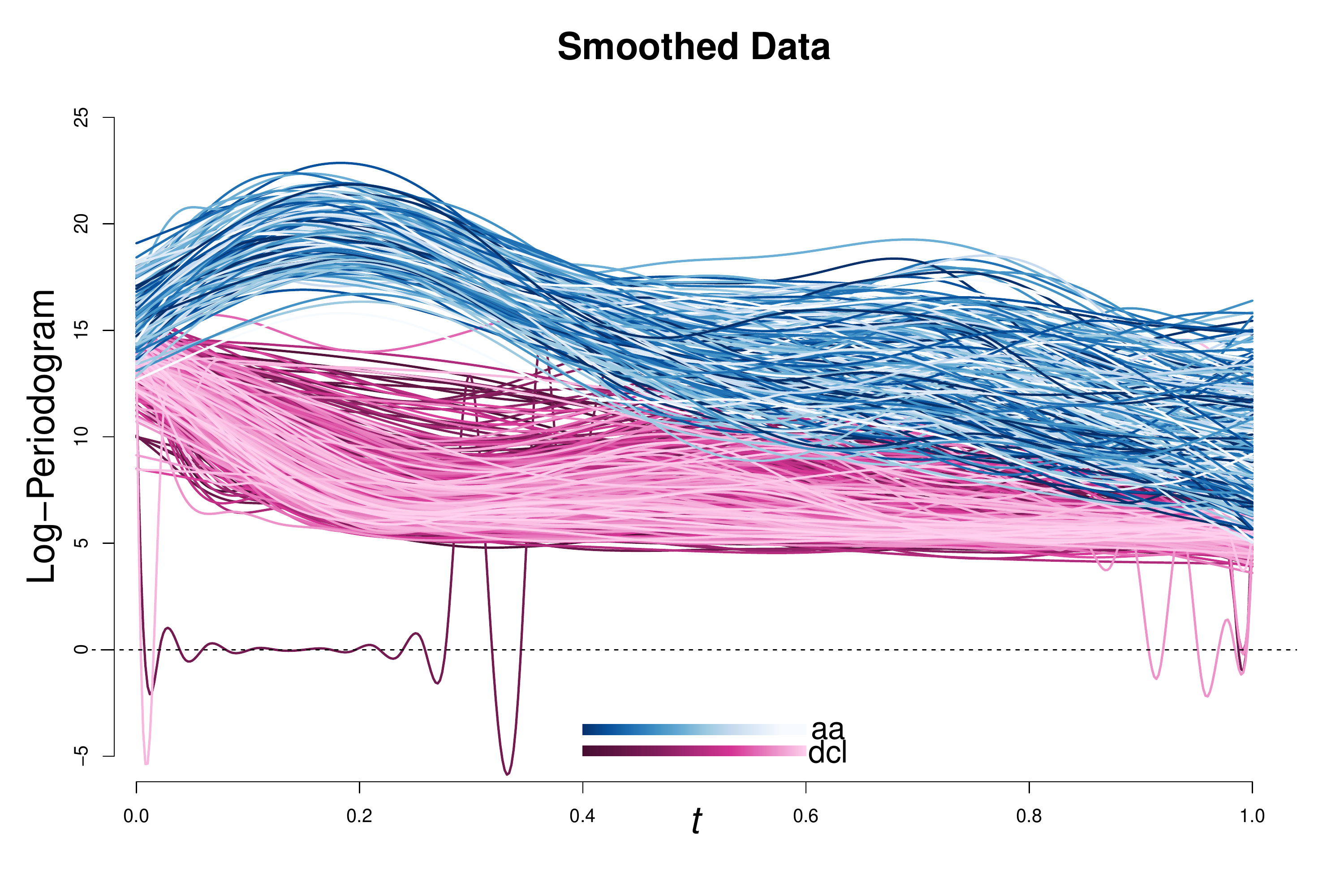}
    \caption*{(a)}
\end{minipage}
\begin{minipage}[c]{.5\textwidth}
        \centering
    \includegraphics[width=\textwidth]{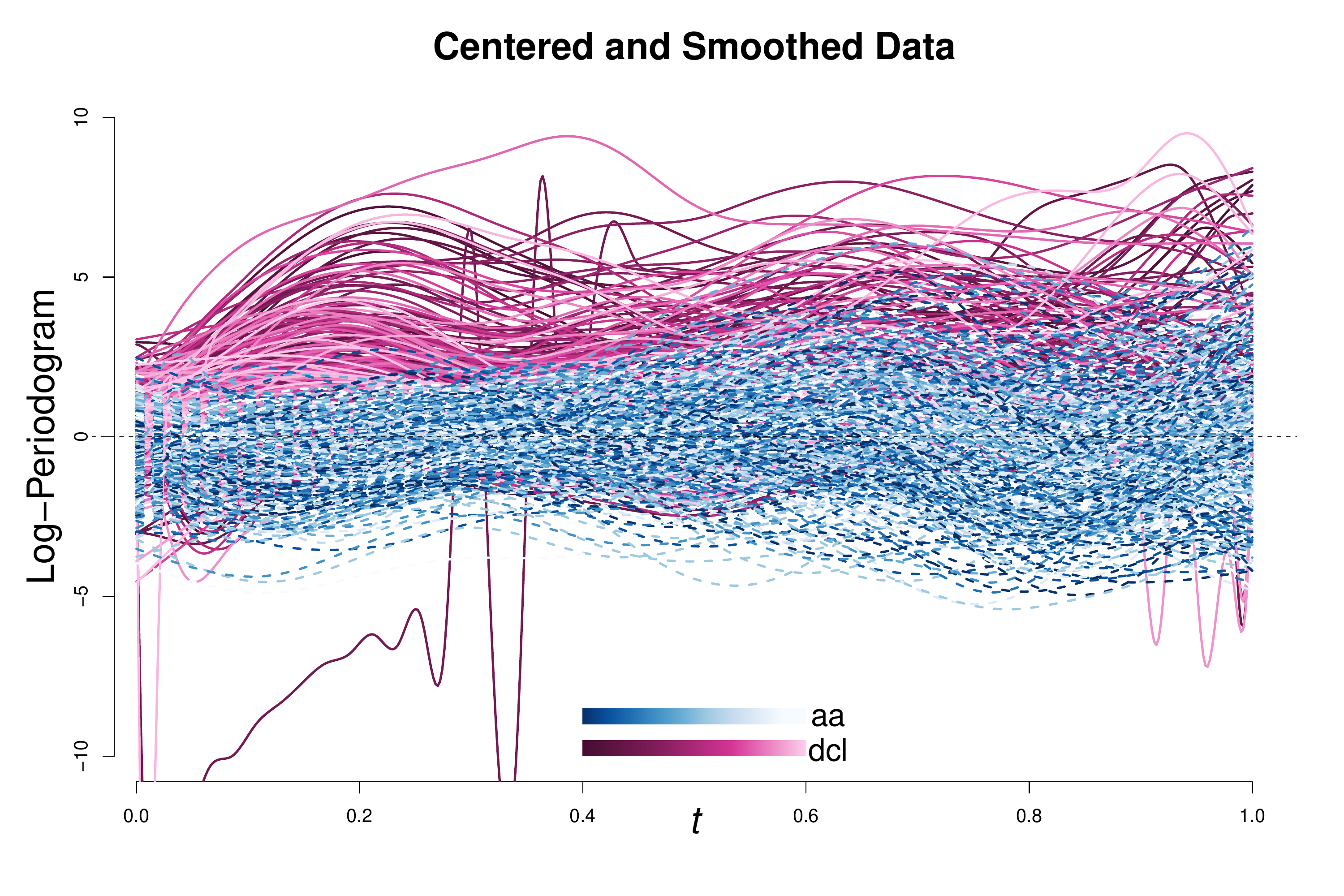}
    \caption*{(b)}
\end{minipage}
\caption{Log periodograms for the syllables `aa' and `dcl'. After centering, we see a difference in magnitude at different times between the curves. There is also a large outlier in the `dcl' group. }
\label{fig::pho_1}
\end{figure}
The data has obvious location differences, for example Figure \ref{fig::pho_1}\textcolor{blue}{(a)} shows the periodograms for two different syllables on the same plot. 
We centered the data by the deepest curve, as measured by random projection depth, within each group. 
We used a robust measure of centre to account for outliers, for example, there is a large outlier in the `dcl' syllable group. 
From Figure \ref{fig::pho_1}\textcolor{blue}{(b)} we might suspect that there are differences in the covariance kernels between the curves. 
To this end, we can run the FKWC test on the five groups of syllables, we run two versions of the test and compare the results. 
We run the FKWC test with random projection depth including the derivative information and as well as the FKWC test with $L^2$-root depth. 
These tests will differ if there are differences of the form $\mathcal{K}_j(s,t)\neq\mathcal{K}_k(s,t)$, $s\neq t,\ j\neq k$. 
Both tests result in incredibly small p-values, smaller than $2.2\times 10^{-16}$. 
We can further examine differences between groups, by performing multiple comparisons. 
Suppose we would like to compare the covariance kernels of groups $j$ and $k$, then there are two obvious routes for multiple comparisons. 
One method is to directly use the pre-calculated joint sample ranks, analogous to the univariate method of \cite{Dunn1964}. 
Here, for large $N$, one is essentially assessing the behaviour of the random variables $\D{X_{j1};F_*}-\D{X_{k1};F_*}$, through combined sample ranks. 
The other method is to extend the methods of \cite{Steel1960} and compare the mean ranks of $\D{X_{ji};F_{jk}}$ and $\D{X_{k\ell};F_{jk}}$, where
$$F_{jk}=\frac{\vartheta_j}{\vartheta_j+\vartheta_k}F_j+\frac{\vartheta_k}{\vartheta_j+\vartheta_k}F_k.$$ 
Since Theorem \ref{thm::SD} and Theorem \ref{thm::RPD} imply that differences in covariance structure will be exhibited in the pairwise ranks of the depth values, when the depth values are taken with respect to the empirical estimate of $F_{jk}$, it seems natural to use the methods of \cite{Steel1960}. 
A second argument in support of the methods of \cite{Steel1960} is as follows. 
Suppose that one group, say $j'\neq k,\ j$ has a very different covariance structure when compared to the remaining groups, if the depth values are computed with respect to the combined sample, then the group $j'$ may then `wash away' any differences between the remaining groups. 
In other words, it is possible that there is a difference between the random variables $\D{X_{ji};F_*}$ and $\D{X_{k\ell};F_*}$, but this difference is small relative to the combined sample and therefore may not be detected. 
The multiple comparisons procedure is as follows. 
For each pair of groups, $j,k$, compute the combined, two-sample depth values: $$\{\D{X_{j1};F_{jk,\scaledN}},\ldots,\D{X_{j\scaled{N}_j};F_{jk,\scaledN}},\D{X_{k1};F_{jk,\scaledN}},\ldots,\D{X_{k\scaled{N}_k};F_{jk,\scaledN}}\},$$
where $F_{jk,\scaledN}$ is the empirical distribution of $\{X_{j1},\ldots,X_{j\scaled{N}_j},X_{k1},\ldots,X_{k\scaled{N}_k}\}$. 
Next, perform the Wilcoxon rank-sum test on the depth values for each pair. 
Lastly, correct the final p-values using the \u{S}id\'{a}k correction \citep{Sidak1967} (or any other multiple testing correction). 
Again, Theorems \ref{thm:::null}, \ref{thm::SD} and \ref{thm::RPD} justify this procedure. 
\begin{table}[t]
\centering
\begin{tabular}{r|rrrrr|rrrrr}
 \multicolumn{1}{c }{ }& \multicolumn{5}{c}{$\RP'$} & \multicolumn{5}{c}{$\LTR$}\\
  \hline
 Syllable & aa  &  ao  &  dcl  &  iy  &  sh  &  aa  &  ao  &  dcl  &  iy  &  sh \\ 
  \hline
aa  & 1.00 & 0.85 & \textbf{0.00} & 0.97 & \textbf{0.00} & 1.00 & 0.13 & 0.40 & 1.00 & \textbf{0.00} \\ 
   ao  & 0.85 & 1.00 & \textbf{0.00} & 0.27 & \textbf{0.00} & 0.13 & 1.00 & \textbf{0.00} & 0.18 & \textbf{0.00} \\ 
   dcl  & \textbf{0.00} & \textbf{0.00} & 1.00 & \textbf{0.00} & \textbf{0.00} & 0.40 & \textbf{0.00} & 1.00 & 0.29 & 0.11 \\ 
   iy  & 0.97 & 0.27 & \textbf{0.00} & 1.00 & \textbf{0.00} & 1.00 & 0.18 & 0.29 & 1.00 & \textbf{0.00} \\ 
   sh  & \textbf{0.00} & \textbf{0.00} & \textbf{0.00} & \textbf{0.00} & 1.00 & \textbf{0.00} & \textbf{0.00} & 0.11 & \textbf{0.00} & 1.00 \\ 
   \hline
\end{tabular}
\caption{\v{S}id\'{a}k corrected p-values of pairwise functional Steel tests performed on the centred curves. }
\label{tab::pho_1}
\end{table}
Table \ref{tab::pho_1} shows \v{S}id\'{a}k corrected p-values of pairwise `functional Steel tests' performed on the centered curves.
The p-values are corrected for the tests done across both hypothesis tests, i.e., across 22 tests. 
We see that the results show that the syllables `dcl' and `sh' differ from the remaining syllables and from each other in terms of the variability of the magnitude of their log-frequencies. 
We can also see that under the $\LTR$ test, `dcl' is similar to the other tests. 
This implies that the trace norm of the covariance operator of `dcl' is similar to that of the other syllables, with the exception of `ao'. 
Under the $\RP'$ test, the syllable `dcl' differs from all other syllables. 
We could interpret this as the log-periodograms of `dcl' are more or less `wiggly' when compared to the other syllables, or that the frequencies that have high variability are different for the syllable `dcl' than the other syllables. 

\bibliography{refs2}
\bibliographystyle{apa}
\appendix
\section{Proofs}\label{sec::app}
\begin{proof}[Proof of Theorem \ref{thm::ltd_prop}]
Let $aF+b$ be the measure associated with $aX+b$. 
We have that 
\begin{align*}
    \LTR(ax+b;aF+b)=\left(1+\E{F}{\norm{ax+b-aX+b}^2}^{1/2}\right)^{-1}=\left(1+\norm{a}\E{F}{\norm{x-X}^2}^{1/2}\right)^{-1}. 
\end{align*}
The function $(1+c'x)^{-1}$ is monotonic for any $c'>0$ and so for any $x,y\in\mathfrak{F}$ such that $\LTR(x;F)<\LTR(y;F)$ then $\LTR(ax+b;aF+b)<\LTR(ay+b;aF+b)$ then $\LTR(ax+b;aF+b)<\LTR(ay+b;aF+b)$. 
This gives the first property. For the second property, observe that
\begin{align*}
    \LTR(x;F)&=\left(1+\E{F}{\norm{x-X}^2}^{1/2}\right)^{-1}\\
    &=\left(1+2^{-1/2}\E{F}{\norm{x-X}^2+\norm{x+X}^2}^{1/2}\right)^{-1}\\
    &=\left(1+2^{-1/2}\E{F}{2\norm{x}^2+2\norm{X}^2}^{1/2}\right)^{-1}\\
    &=\left(1+2^{-1/2}2^{1/2}\norm{x}^2+2^{-1/2}2^{1/2}\E{}{\norm{X}^2}^{1/2}\right)^{-1}\\
    &=\left(1+\norm{x}+c'\right)^{-1},
\end{align*}
which is maximized at $x=\mathbf{0}$. 
For the third and fourth properties, we have in similar fashion
\begin{align*}
   \LTR(cx;F)&=\left(1+\E{F}{\norm{cx-X}^2}^{1/2}\right)^{-1}=\left(1+c\norm{x}+c'\right)^{-1},
\end{align*}
which is decreasing toward 0 as $c$ increases. 
Lastly, if $X_1,\ldots,X_N$ is a random sample from $F$, then it holds that
$$\frac{1}{N}\sum_{i=1}^N\norm{x-X_i}^2=\norm{x}^2+\frac{1}{N}\sum_{i=1}^N\norm{X_i}^2-2\ip{x,\frac{1}{N}\sum_{i=1}^N X_i}\coloneqq\norm{x}^2+\overline{Y}_{x,\scaledN}.$$
We have that 
\begin{align*}
|\LTR(x;F_{\scaledN})-\LTR(x;F)|    &=\left|\left(1+(\norm{x}^2+\overline{Y}_{x,\scaledN})^{1/2}\right)^{-1}-\left(1+(\norm{x}^2+\norm{\kerrO}_{TR})^{1/2}\right)^{-1}\right|\\
 &=\left|\frac{(\norm{x}^2+\norm{\kerrO}_{TR})^{1/2}-(\norm{x}^2+\overline{Y}_{x,\scaledN})^{1/2}}{\left(1+(\norm{x}^2+\overline{Y}_{x,\scaledN})^{1/2}\right)\left(1+(\norm{x}^2+\norm{\kerrO}_{TR})^{1/2}\right)}\right|\\
 &\leq\left|\frac{|\norm{\kerrO}_{TR}-\overline{Y}_{x,\scaledN}|^{1/2}}{\left(1+(\norm{x}^2+\overline{Y}_{x,\scaledN})^{1/2}\right)\left(1+(\norm{x}^2+\norm{\kerrO}_{TR})^{1/2}\right)}\right|,
 \end{align*}
 where the third line comes from the fact that $\sqrt{x}-\sqrt{y}\leq \sqrt{|x-y|}$. 
 Now, suppose that $\norm{x}<c'N^{1/2}(\log N)^{-1}$. 
\begin{align*}
 |\LTR(x;F_{\scaledN})-\LTR(x;F)|   &\leq |\norm{\kerrO}_{TR}-\overline{Y}_{x,\scaledN}|^{1/2}\\
 &\leq \left|\norm{\kerrO}_{TR}-\frac{1}{N}\sum_{i=1}^N\norm{X_i}^2-c'N^{-1/2}(\log N)^{-1}\sum_{i=1}^N\int X_i dt \right|^{1/2}\\
 &=o(1)\ a.s.\ ,
\end{align*}
where the last line is from the strong law of large numbers and the law of the iterated logarithm. 
Note that $\int X_i dt$ has finite variance for all $i\in\{1,\ldots ,N\}$ and that the second inequality does not depend on $x$. 
Suppose now that $\norm{x}\geq c'N^{1/2}(\log N)^{-1}$. 
Then, it is easy to see that $|\LTR(x;F_{\scaledN})-\LTR(x;F)|\rightarrow 0$ by the vanishing at infinity property. 
\end{proof}
\begin{proof}[Proof of Theorem \ref{thm:::null}]
Observe that $\widehat{R}_{ji}$ are identically distributed under the null hypothesis. 
This implies that the rank vector has the uniform distribution with probability of each outcome being $1/N!$. 
This is the same setup as in \citep{Kruskal1952} and so it is immediate that $\mathcal{W}_{\scaledN}\cond \rchi^2_{\scaled{J}-1}$. 
Similarly, it follows directly from Theorem 2 of \citep{Chenouri2011} that $\mathcal{M}_{\scaledN,r}\cond \rchi^2_{\scaled{J}-1}$. 
\end{proof}
\begin{proof}[Proof of Theorem \ref{thm:::alt}]
Let $$\widetilde{\sigma}^2_{\scaledN}=\frac{N(N+1)}{12}.$$
We can rewrite $\widehat{\mathcal{W}}_{\scaledN}$ as follows
\begin{equation*}
    \widehat{\mathcal{W}}_{\scaledN}=\frac{1}{\widetilde{\sigma}^2_{\scaledN}}\sum_{j=1}^{J} N_j \rhatbarj^{2}-3(N+1). 
\end{equation*} 
Now, define
\begin{align}
 \label{eqn:T}
   \mathcal{W}_{\scaledN}&\coloneqq \frac{1}{\widetilde{\sigma}^2_{\scaledN}} \sum_{j=1}^{J} N_j \overline{R}_{j}^{2}-3(N+1)\qquad \text{with}\qquad  \overline{R}_{j}\coloneqq\frac{1}{N_j}\sum_{i=1}^{N_j} R_{ji}.
\end{align}
Under the alternative hypothesis, Assumption \ref{ass:cont} and Assumption \ref{ass:med_d}, $\D{X_{ji};F_*}$ are equivalent to the univariate random variables studied by \cite{Kruskal1952}. 
For all $\delta>0$, it then holds that
$$\mathrm{P}\left(\mathcal{W}_{N}>\delta\right) \rightarrow 1,\ \text { as } N \rightarrow \infty.$$
We now show that $|\widehat{\mathcal{W}}_{\scaledN}-\mathcal{W}_{\scaledN}|=O_p(1),$ which will complete the proof. 
To this end, note that the sequences $R_{j1},\ldots, R_{j\scaledN_j}$ and $\widehat{R}_{j1},\ldots, \widehat{R}_{j\scaledN_j}$ are both triangular arrays of exchangeable random variables. 
This fact allows us to apply the central limit theorem of \cite{Weber1980}. 
Specifically, it holds that
\begin{align} 
\frac{\sqrt{N_j}}{\Var{R_{j1}}}(\overline{R}_{j}-\E{}{R_{j1}})=O_p(1) \qquad\text{ and }\qquad \frac{\sqrt{N_j}}{\Varr{\widehat{R}_{j1}}}\left(\overline{\widehat{R}}_{j}-\Eee{}{\widehat{R}_{j1}}\right)=O_p(1).
\label{eqn::weberclt}
\end{align}
We now relate these two quantities with $\widehat{R}_{ji}=R_{ji}+\mathcal{E}_{ji}$ where,
\begin{align}
\label{eqn::rankrep}
    \mathcal{E}_{ji}&=\sum_{\ell=1}^J\sum_{m=1}^{N_j}\ind{B_{{ji},\ell m}}-\ind{A_{{ji},\ell m}},
\end{align}
and
\begin{align*} A_{ji, \ell m} &=\left\{D\left(X_{ji}, F_{*}\right) \leq D\left(X_{\ell m}, F_{*}\right)\right\} \cap\left\{D\left(X_{ji}, F_{*,\scaleto{N}{4pt}}\right)>D\left(X_{\ell m}, F_{*,\scaleto{N}{4pt}}\right)\right\} \\ B_{ji, \ell m} &=\left\{D\left(X_{ji}, F_{*}\right)>D\left(X_{\ell m}, F_{*}\right)\right\} \cap\left\{D\left(X_{ji}, F_{*,\scaleto{N}{4pt}}\right) \leq D\left(X_{\ell m}, F_{*,\scaleto{N}{4pt}}\right)\right\}.
\end{align*}
Observe that $\E{}{\ind{B_{ji,\ell m}}}=O(N^{-1/2})$ and $\E{}{\ind{A_{ji,\ell m}}}=O(N^{-1/2})$, which is due to Assumption \ref{ass:cont} and Assumption \ref{ass:consDepth}, \citep[see the paragraph between A4 and A5 of][where one notes that Assumption \ref{ass:cont} implies that $\Pr(|\D{X_{j1};F_{*}}-\D{X_{k1};F_{*}}|\leq v)$ is also Lipschitz in $v.$]{Chenouri2020DD}. 
It then follows that 
$\Eee{}{\mathcal{E}_{ji}}=\Eee{}{\widehat{R}_{ji}}-\Eee{}{R_{ji}}=O(N^{1/2}).$
With this fact in mind, we now show that 
\begin{equation}
\label{eqn::varsame}
    \Varr{\widehat{R}_{ji}}/\Varr{R_{ji}}=O(1)\qquad\text{ and }\qquad \Varr{R_{ji}}/\widetilde{\sigma}^2_{\scaleto{N}{4pt}}=O(1).
\end{equation}
The right-side identity follows easily from the fact that $N_j/N=\vartheta_j+o(1)$; $\Varr{R_{ji}}=O(N^2),$ for any $j\in\{1,\ldots,J\}$ and $i\in \{1,\ldots,N_j\}$. 
For the left identity, we can write
\begin{align*}
\Var{\widehat{R}_{ji}}&=\Var{R_{ji}+\mathcal{E}_{ji}}\\
&=\Var{R_{ji}}+\Var{\mathcal{E}_{ji}}+2\,\mathrm{C}ov\left(\mathcal{E}_{ji},R_{ji}\right)\\
&\leq \Var{R_{ji}}+\Var{\mathcal{E}_{ji}}+2\,\E{}{|\mathcal{E}_{ji}-\E{}{\mathcal{E}_{ji}}|}N\\
&= \Var{R_{ji}}+\Var{\mathcal{E}_{ji}}+O(N^{3/2})\\
&=\Var{R_{ji}}+\E{}{\left(\sum_{\ell=1}^J\sum_{m=1}^{N_j}\ind{B_{{ji},\ell m}}-\ind{A_{{ji},\ell m}}\right)^2}+O(N)+O(N^{3/2})\\
&\leq\Var{R_{ji}}+\E{}{\sum_{\ell=1}^J\sum_{m=1}^{N_j}\ind{B_{{ji},\ell m}}+\ind{A_{{ji},\ell m}}}+O(N)+O(N^{3/2})\\
&\leq \Var{R_{ji}}+O(N^{3/2}),
\end{align*}
where the fourth line comes from applying equation (A5) of \cite{Chenouri2020DD}. 
The last line is from the the fact that $\E{}{\ind{B_{ji,\ell m}}}=O(N^{-1/2})$ and $\E{}{\ind{A_{ji,\ell m}}}=O(N^{-1/2})$. 
Now, 
$$\limn \frac{\Varr{\widehat{R}_{ji}}}{\Varr{R_{ji}}}=\limn\frac{\Varr{\widehat{R}_{ji}}/N^2}{\Varr{R_{ji}}/N^2}=\limn\frac{\Varr{R_{ji}}/N^2+o(1)}{\Varr{R_{j1}}/N^2}=1.$$
It then follows from Slutsky's theorem, continuous mapping theorem and the central limit theorem of \cite{Weber1980} that
\begin{align}
    \widehat{\mathcal{W}}_{\scaledN}-\mathcal{W}_{\scaledN}&=\frac{1}{\widetilde{\sigma}^2_{\scaleto{N}{4pt}}}\sumJ N_j\left(\overline{\widehat{R}}_{j}^{2}-\overline{R}_{j}^{2}\right)\nonumber\\
    &= \sumJ \left[\bigg( \frac{\sqrt{N_j}\ \overline{\widehat{R}}_{j}}{\widetilde{\sigma}_{\scaleto{N}{4pt}}}\bigg)^2-\bigg( \frac{\sqrt{N_j}\ \overline{R}_{j}}{\widetilde{\sigma}_{\scaleto{N}{4pt}}}\bigg)^2\right]\nonumber\\
    &=O_p(1)+ \frac{1}{\widetilde{\sigma}^2_{\scaleto{N}{4pt}}}\sumJ N_j\left[ \Eee{}{\widehat{R}_{ji}}^2-\Eee{}{R_{ji}}^2+\Eee{}{R_{ji}}\overline{R}_{j}-\Eee{}{\widehat{R}_{ji}}\overline{\widehat{R}}_{j}\right]\nonumber\\
    &=O_p(1).\nonumber
\end{align}
The third line results from we adding and subtracting $\frac{\sqrt{N_j}\ \mathrm{E}[\widehat{R}_{ji}]}{\widetilde{\sigma}_{\scaleto{N}{4pt}}}$ and $\frac{\sqrt{N_j}\ \E{}{R_{ji}}}{\widetilde{\sigma}_{\scaleto{N}{4pt}}}$ inside the left and right round brackets respectively, multiplying out and lastly applying \eqref{eqn::weberclt} and \eqref{eqn::varsame}.
\end{proof}
\begin{proof}[Proof of Theorem \ref{thm::SD}]
In view of Theorem \ref{thm:::alt}, it is only necessary to show that under the alternative, Assumption \ref{ass:med_d} holds. 
We take univariate observations with equivalent ranks to the depth values. 
The function $(1+a^{1/2})^{-1}$ is a monotonic function in $a$, which implies that we can use ranks based on $\E{}{\norm{X_{ji}-Z}^2\big|X_{ji}},$ where $Z\sim F_*$.  
We can then write
$$\widetilde{Y}_{ji}=\E{}{\norm{X_{ji}-Z}^2\big|X_{ji}}=\norm{X_{ji}}^2+\E{}{\norm{Z}^2}-\E{}{\,2\,\ip{X_{ji},Z}\,|\,X_{ji}\,}=\norm{X_{ji}}^2+\norm{\kerrO_Z}_{TR},$$
where $\kerrO_Z$ is the covariance operator corresponding to $Z$ and the last equality comes from the fact that $\E{}{Z}=\mathbf{0}$. 
We can re-center $\widetilde{Y}_{ji}$ by $\norm{\kerrO_Z}_{TR}$ to finally use ranks generated by $Y_{ji}=\norm{X_{ji}}^2.$ 
Therefore, we have that 
$$\E{}{\norm{X_{11}}^2-\norm{X_{21}}^2}=\norm{\kerrO_1}_{TR}-\norm{\kerrO_2}_{TR}\neq 0. $$
\end{proof}
\begin{proof}[Proof of Theorem \ref{thm::RPD}]
First, let $u\in S$, where $S=\{u\colon \norm{u}=1,\ u\in\mathfrak{F}\}$ and let $Y_{u,j}=\ip{X_{j1},u}$. 
Observe that $\E{}{Y_{u,j}}=0$ and that
\begin{align*}
    \sigma^2_{j,u}&\coloneqq \E{}{Y^2_{u,j}}=\E{}{\intd\intd  X_{j1}(t)u(t)\cdot X_{j1}(s)u(s)\ dsdt}=\ip{\kerrO_j u,u},
\end{align*}
where we can take the expectation inside due to Lebesgue's dominated convergence theorem. 
Namely, $\ip{X_{j1},u}^2\leq \norm{X_{j1}}^2$ which has finite expectation. 
One should also recall that in $\RP_\infty$ we take $$\D{\ip{x,u};F_u}=F_u(\ip{x,u})(1-F_u(\ip{x,u}))$$ for the univariate depth. 
For the remainder of the proof we suppress the $\infty$ in $\RP_\infty$ for brevity. 
Now, following the same argument from the first paragraph of the proof of Theorem \ref{thm::SD}, it is only necessary to verify that 
\begin{equation}
    \Pr\left(\RP(X_{11};F^*)>\RP(X_{21};F_{*})\right)\neq \frac{1}{2}.
\end{equation}
Which, under the conditions of the theorem, this is equivalent to showing 
$$\E{}{\RP(X_{11};F_{*})-\RP(X_{21};F_{*})}\neq 0.$$
We can write
\begin{align*}
\E{}{\RP(X_{11};F_{*})-\RP(X_{21};F_{*})}&=\E{}{\int_{\mathcal{C}}\D{Y_{u,1};F_{u,*}} d\nu(u)-\int_{\mathcal{C}}\D{Y_{u,2};F_{u,*}} d\nu(u)}\\
&=\E{}{\int_{\mathcal{C}}[ F_{u,*}(Y_{u,1})(1-F_{u,*}(Y_{u,1}))-F_{u,*}(Y_{u,2})(1-F_{u,*}(Y_{u,2}))]\ d\nu(u)}.
\end{align*}
Clearly, since $0<F_{u,*}(Y_{u,1})<1$,  we have that
$$ F_{u,*}(Y_{u,1})(1-F_{u,*}(Y_{u,1}))-F_{u,*}(Y_{u,2})(1-F_{u,*}(Y_{u,2}))\leq 1/2.$$
Using Lebesgue's dominated convergence theorem,
\begin{align*}
\E{}{\RP(X_{11};F^*)-\RP(X_{21};F^*)}&=\int_{\mathcal{C}}\E{}{ F_{u,*}(Y_{u,1})(1-F_{u,*}(Y_{u,1}))-F_{u,*}(Y_{u,2})(1-F_{u,*}(Y_{u,2}))}d\nu(u).
\end{align*}
Using the fact that $F_{*,u}$ is thrice differentiable for all $u$, we can write 
\begin{align*}
\E{}{F_{u,*}(Y_{u,j})}&=F_{u,*}(0)+\frac{1}{2}f_{u,*}^{(1)}(0)\sigma_{j,u}^2+\mathcal{R}_{u,j,1} \\
\E{}{F^2_{u,*}(Y_{u,j})}&=F^2_{u,*}(0)+(F_{u,*}(0)f_{u,*}^{(1)}(0)+f^2_{u,*}(0))\sigma_{j,u}^2+\mathcal{R}_{u,j,2} \intertext{ with } 
\mathcal{R}_{u,j,1}&\coloneqq \E{}{\frac{1}{6}\int_0^{Y_{u,j}}f_{u,*}^{(2)}(t)(Y_{u,j}-t)^3dt}\\
\mathcal{R}_{u,j,2}&\coloneqq \E{}{\frac{1}{3}\int_0^{Y_{u,j}}(3f_{u,*}(t)f_{u,*}^{(1)}(t)+F_{u,*}(t)f_{u,*}^{(2)}(t))(Y_{u,j}-t)^3dt}.
\end{align*}
Note that we expect $\mathcal{R}^i_{u,j}$ to be small from 
the fact that the mean of $Y_{u,j}$ is 0. 
It follows that
\begin{align*}
\E{}{\D{Y_{u,j};F_{u,*}}}&=F_{u,*}(0)+\frac{1}{2}f_{u,*}^{(1)}(0)\sigma_{j,u}^2+\mathcal{R}_{u,j,1}-F^2_{u,*}(0)-(F_{u,*}(0)f_{u,*}^{(1)}(0)-f^2_{u,*}(0))\sigma_{j,u}^2-\mathcal{R}_{u,j,2}\\
&=\mathcal{H}(F_{*,u})\sigma_{j,u}^2+F_{u,*}(0)-F^2_{u,*}(0)+\mathcal{R}_{u,j,3},
\end{align*}
where 
$$\mathcal{H}(F)\coloneqq \frac{1}{2}f^{(1)}(0)-(F(0)f^{(1)}(0)-f^2(0))\qquad\text{and}\qquad \mathcal{R}_{u,j,3}=\mathcal{R}_{u,j,1}-\mathcal{R}_{u,j,2}.$$
We can now write
\begin{align*}
\E{}{\D{Y_{u,1};F_{u,*}}-\D{Y_{u,2};F_{u,*}}}&=\mathcal{H}(F_{*,u})(\sigma_{1,u}^2-\sigma_{2,u}^2)+\mathcal{R}_{u,1,3}-\mathcal{R}_{u,2,3}.
\end{align*}

To conclude, under univariate simplicial depth it holds that
\begin{align*}
    \E{}{\RP(X_{11};F^*)-\RP(X_{21};F^*)}&=\int_{\mathcal{C}}\mathcal{H}(F_{*,u})(\sigma_{1,u}^2-\sigma_{2,u}^2)+\mathcal{R}_{u,1,3}-\mathcal{R}_{u,2,3}\ d\nu(u)\\
    &=\int_{\mathcal{C}}\mathcal{H}(F_{*,u})\ip{\kerrO_1 u-\kerrO_2 u,u}+\mathcal{R}_{u,1,3}-\mathcal{R}_{u,2,3}\ d\nu(u)\\
    &=\int_{\mathcal{C}}\mathcal{H}(F_{*,u})\ip{\kerrO_1 u-\kerrO_2 u,u} d\nu(u)+\mathcal{R}_{1}-\mathcal{R}_{2},
\end{align*}
where $\mathcal{R}_{j}<\infty$ by the fact that the integrand is bounded in $u$. 
\end{proof}
\begin{proof}[Proof of Theorem \ref{thm::sdla}]
In view of the proof of Theorem \ref{thm::SD}, we can use the fact that the $\LTR$ depth-based ranks are equivalent to ranks generated by $Y_{ji}=\norm{X_{ji}}^2.$ 
Now, $Y_{ji}$ are univariate observations from a scale family, meaning that $Z_{ji}\coloneqq \left(1+\delta_j/\sqrt{N}\right)^{-1}Y_{ji}$ with $Z_{ji}\sim G$. 
Now, let $\tau\coloneqq \limN\tau_{\scaledN}$. 
It follows from \citep{Fan2011} that the test statistic $\mathcal{W}_{\scaledN}\rightarrow \rchi_{J-1}^2(\tau)$. 
We have that
\begin{equation*}
    \tau_{\scaledN}=\frac{12}{N(N+1)}  \sum_{j=1}^{J} N_{j}\left\{N \sum_{k \neq j} \vartheta_{k}\left(\Pr(Y_{k1}\leq Y_{j1})-1 / 2\right)\right\}^{2}.
\end{equation*}
We have that
\begin{align}
    \Pr(Y_{k1}\leq Y_{j1})&=\Pr\left(Z_{k1}\leq Z_{j1}\left[\frac{\sqrt{N}+\delta_j}{\sqrt{N}+\delta_k}\right]\right)=\bigintsss_{\re}\Pr\left(Z_{k1}\leq z\left[\frac{\sqrt{N}+\delta_j}{\sqrt{N}+\delta_k}\right]\right)g(z)dz.
    \label{eqn::subbb}
\end{align}
Now, note that $z\left[\frac{\sqrt{N}+\delta_j}{\sqrt{N}+\delta_k}\right]$ is in a neighborhood of $z$ we can write the Taylor expansion of $G$ about $z$ at the point $z\left[\frac{\sqrt{N}+\delta_j}{\sqrt{N}+\delta_k}\right]$ as
\begin{align*}
    \Pr\left(Z_{k1}\leq z\left[\frac{\sqrt{N}+\delta_j}{\sqrt{N}+\delta_k}\right]\right)&=G(z)+z\left[1-\left[\frac{\sqrt{N}+\delta_j}{\sqrt{N}+\delta_k}\right]\right]g(z)+O(N^{-1}).
\end{align*}
Substituting into \eqref{eqn::subbb}, we have that 
\begin{align*}
   \bigintsss_{\re}\Pr\left(Z_{k1}\leq z\left[\frac{\sqrt{N}+\delta_j}{\sqrt{N}+\delta_k}\right]\right)g(z)dz &=\bigintsss_\re\left[ G(z)+z\left[1-\left[\frac{\sqrt{N}+\delta_j}{\sqrt{N}+\delta_k}\right]\right]g(z)+O(N^{-1})\right]dG(z)\\
   &=1/2+\left[1-\left[\frac{\sqrt{N}+\delta_j}{\sqrt{N}+\delta_k}\right]\right]\bigintsss_\re zg(z)^2dz +O(N^{-1})\\
   &= 1/2+\left[1-\left[\frac{\sqrt{N}+\delta_j}{\sqrt{N}+\delta_k}\right]\right]\Delta_{G} +O(N^{-1}),
\end{align*}
where $$\Delta_{G} =\bigintsss_\re zg(z)^2dz.$$
Now, substituting the above identity into \eqref{eqn::tau}, gives
\begin{align*}
    \tau_{\scaledN}&=\frac{12}{N(N+1)} \sum_{j=1}^{J} N_{j}\left[N \sum_{k \neq j} \vartheta_{k}\left(\left[1-\left[\frac{\sqrt{N}+\delta_j}{\sqrt{N}+\delta_k}\right]\right]\Delta_{G} +O(N^{-1})\right)\right]^{2},
\end{align*}
which then immediately implies that
   $$ \limN \tau_{\scaledN}=12 \Delta_{G}^2  \sum_{j=1}^{J} \vartheta_{j}\left( \delta_j-\overline{\delta}\right)^{2},$$
which completes the proof. \qedhere
\end{proof}

\end{document}